\documentclass[12pt,a4paper]{article}

\usepackage{amsmath}
\usepackage{amsthm}
\usepackage{amssymb}
\usepackage[utf8]{inputenc}
\usepackage{fancyhdr}

\newtheorem{theorem}{Theorem}
\newtheorem{lemma}[theorem]{Lemma}
\newtheorem{corrol}[theorem]{Corollary}

\def\Real{{\mathbb R}}

\def\cnj#1{{\overline{#1}}}

\def\innerprod(#1,#2){{\left<#1\,,\,#2\right>}}
\def\Set#1{{\left\{#1\right\}}}
\def\qquadtext#1{\qquad\textup{#1}\qquad}
\def\qquadand{\qquadtext{and}}
\def\quadtext#1{\quad\textup{#1}\quad}
\def\quadand{\quadtext{and}}
\def\pfrac#1#2{\frac{\partial #1}{\partial #2}}

\def\dfrac#1#2{\frac{d #1}{d #2}}

\def\dual#1{{\widetilde{#1}}}
\def\Lie{{\cal L}}
\def\Ebun{{\cal E}}
\def\Amap{{\cal A}}
\def\Bmap{{\cal B}}
\def\Tprod{{\colon\!}}
\def\TprodL{{\vdots}}

\def\LSym{{{\cal B}}}
\def\Met{{{\cal H}}}
\def\Ein{{\textup{Ein}}}

\def\Noet{{\cal N}}
\def\tenbun{{\boldsymbol \otimes}}
\def\Tform{{\mathbb{F}}}
\def\Tvec{{\mathbb{V}}}
\def\slist{{\boldsymbol{s}}}
\def\tlist{{\boldsymbol{t}}}
\def\Zalt{{\cal Z}}
\def\NumZ{{N}}
\def\Numzeta{{Q}}
\def\GatZ#1{\frac{\Delta \LambdaM}{\Delta {#1}}}
\def\Vasym{{\cal V}}
\def\Zind{A}
\def\ZZind{C}
\def\zeind{B}
\def\Mat{{\mathbf{m}}}
\def\LambdaM{\Lambda^{\Mat}}
\def\LambdaT{{\Lambda^{\mathbf{T}}}}
\def\MST{{\cal M}}
\def\CR{{\boldsymbol \chi}}
\def\Tden{{\cal T}}
\def\Sden{{\cal S}}
\def\calH{{\cal H}}
\def\CRpre{{\kappa}}
\def\Rot{{\cal R}}
\def\MoM{{\cal M}}

\title{Conservation Laws and Stress-energy-momentum Tensors for
  Systems with Background Fields}

\author{Jonathan Gratus\thanks{j.gratus@lancaster.ac.uk, Corresponding
    Author},
\\
\small{Lancaster University, Lancaster LA1 4YB, UK 
and}\\ \small{The Cockcroft Inistitute, Daresbury Laboratory, Warrington WA4
4AD, UK}
\\\\
Yuri N Obukhov\thanks{yo@thp.uni-koeln.de},
\\
\small{Institute for Theoretical Physics, University of Cologne,
50923 K\"oln, Germany}
\\\\
Robin W Tucker\thanks{r.tucker@lancaster.ac.uk}
\\
\small{Lancaster University, Lancaster LA1 4YB, UK 
and}\\ \small{The Cockcroft Inistitute, Daresbury Laboratory, Warrington WA4
4AD, UK}
}

\begin{document}
\maketitle

\thispagestyle{fancy}
\rhead{Cockcroft-12-26}

\begin{abstract}
  This article attempts to delineate the roles played by non-dynamical
  background structures and Killing symmetries in the construction of
  stress-energy-momentum tensors generated from a diffeomorphism
  invariant action density. An intrinsic coordinate independent
  approach puts into perspective a number of spurious arguments that
  have historically lead to the main contenders, viz the
  Belinfante-Rosenfeld stress-energy-momentum tensor derived from a
  Noether current and the Einstein-Hilbert stress-energy-momentum
  tensor derived in the context of Einstein's theory of general
  relativity.  Emphasis is placed on the role played by non-dynamical
  background (phenomenological) structures that discriminate between
  properties of these tensors particularly in the context of
  electrodynamics in media. These tensors are used to construct
  conservation laws in the presence of Killing Lie-symmetric
  background fields.
\end{abstract}

\noindent{\bf keywords}
Noether currents, 
Abraham stress-energy-momentum tensor, \\
Minkowski stress-energy-momentum tensor, \\
Canonical stress-energy-momentum tensor, \\
Einstein-Hilbert stress-energy-momentum tensor,\\  
Diffeomorphism invariance







\section{Introduction}
\label{ch_Intro}
The venerable aim of the variational calculus as a tool for deriving
from a single principle the observed laws of physics remains an
attractive one. It synthesises conservation laws from local symmetries
and offers a route for finding unified schemes that
underpin the experimental sciences. However the epistemology that
arises in this approach has, on occasion, lead to unnecessary confusion
when employed in the discussion of systems that require
phenomenological input. In particular this remains true in the context
of electrodynamics in macroscopic material where the indiscriminate use of the
stress-energy-momentum concept in classical physics has lead to
different notions of quantum electrodynamics in such media.
The calculus often enables observed local and
global symmetries of nature to be accommodated in a natural manner and
can be used to generate conserved currents associated with local
symmetries. Furthermore the local dynamics of coupled systems follows
from local extrema of a functional constructed from a single
diffeomorphism invariant action integral $S=\int_{\MST} \LambdaT$,
for some diffeomorphism invariant differential 4-form $\LambdaT$ on
spacetime $\MST$.  If one chooses a nowhere vanishing 4-form $\Omega$
on spacetime and writes $\LambdaT={\Sden} \Omega$ then ${\Sden}$ is
called a scalar action density (relative to $\Omega$).  For any
covariant metric $g$ on $\MST$ a natural choice is $\Omega=\sqrt{
  \vert {\text{det}}g \vert}\, d^4 x=\star 1$ where $\star$ is the
Hodge map associated with $g$. An alternative choice, given a
coordinate system $(x^1,\ldots,x^4)$ is the {\em tensor density}
$\Omega=dx^1\wedge \cdots \wedge dx^4$. Such a choice implies that
$\Sden$ then transforms like $\vert{\text{det}}g \vert^{-1/2}$
under a local change of coordinates\footnote{If $U=\partial/\partial {x^1}$
generates time translations then $i_U \LambdaT$ may be called a
Lagrangian 3-form \cite{goldstein} relative to $U$.}.

Traditional symmetry analysis based on the pioneering work of Noether
exploits such densities often in a manifestly coordinate dependent
manner. Modern approaches exploit an intrinsic jet-bundle
formulation. Between these extremes are ad-hoc formulations that have
led to intense debates about the role of alternative
stress-energy-momentum tensors used to describe the interchange of
energy and momentum between interacting systems in continuous
macroscopic media. In particular, conserved Noether currents have
historically been constructed from 3-forms
on space and their associated Lagrangian scalar densities. Notions of {\it stress} and {\it power} in continuum mechanics became incorporated into the Lorentz covariant formulation of special relativity.  With the
advent of Einstein's general theory of relativity the spacetime
symmetric metric tensor field becomes a dynamical variable and an
alternative stress-energy-momentum tensor can be defined through the
variational derivative of the action 4-form on spacetime (for gravity
and matter) with respect to this metric tensor.\footnote{Matter in
  this sense includes electromagnetism as well as other matter fields.
  Modified theories of gravity, e.g. with scalar fields coupled
  directly to spacetime curvature, offer no unique identification of a
  ``matter'' action.}  The issue that then arises is how best to
identify conserved quantities constructed from different models that
can be put into a variational formulation and relate them to different
choices of stress-energy-momentum tensor.

 Conserved quantities can be
generated from closed 3-forms in spacetime. These often arise as a
consequence of some symmetry of an element or elements contributing to
the structure of the action for the model. However even in the absence
of such symmetries there are powerful relations that arise from
applying the variational approach to action functionals that are
(locally) invariant under transformations of the fields induced by
arbitrary (local) spacetime diffeomorphisms. Such functionals are
readily constructed in terms of coordinate invariants made by
contracting tensor fields of various degrees or tensor ``densities''
of different weights. It is not necessary that
all such quantities in this construction be dynamical, i.e. subject to
variational field equations. However it is in these circumstances that
there arise differences when one compares the consequences deduced
from different choices for the description of stress, energy, and
momentum in the presence of such quantities.

The (covariant) variational approach is often restricted to {\it
  closed} (non-dissipative) systems where dynamical equations for the
dynamical field variables arise by finding local extrema of a total
action $S$ under their variations, where $\LambdaT$ is
a 4-form on a spacetime manifold. Such forms belong to a class,
members of which describe the same classical physics. For example
members that differ by an exact 4-form with compact support
yield the same variational equations.

Since the description of gravitation is given in terms of a geometry
of spacetime, any collection of dynamical field variables
$\{\boldsymbol\zeta\}=\Set{\zeta_1,\ldots,\zeta_\Numzeta}$ that are
not part of the geometry may be assigned the status of
``matter''.\footnote {In view of the ambitions of M-theory this
  interpretation may be an effective one imposed by dimensional
  reduction.}  In Einstein's theory, gravitation arises from a
pseudo-Riemannian spacetime geometry based entirely on a dynamical
spacetime metric tensor field $g$. In the absence of matter, the
dynamics of $g$ is given by the Einstein-Hilbert action $\int_\MST
\Lambda^{\Ein}$.  The dynamics of (non-spinorial) matter minimally
coupled to gravitation is described by a matter action 4-form
$\LambdaM$ that is independent of derivatives of $g$.

In order to eliminate use of the jet-bundle language and simplify the
analysis, in the following the matter action 4-form $\LambdaM$ will be
restricted to depend on $g$, $\{{\boldsymbol Z}\}$,
$\{{\boldsymbol\zeta}\}$ and $\{d{\boldsymbol\zeta}\}$, where a
collection of prescribed non-dynamical background tensor fields $
\{\boldsymbol Z\}=\Set{Z_1,\ldots,Z_\NumZ}$ has been included.  Such fields play no variational role in
fixing the local extrema of $S$ but (if present) play a crucial role
in determining the consequences that follow from diffeomorphism
invariance.  Then the total action 4-form is $\LambdaT =
\Lambda^{\Ein} + \LambdaM$.

It is sometimes useful to define matter {\em subsystems} that are
described by sub-actions $\int_\MST \Lambda^s$ where
$\LambdaM=\sum_s\Lambda^s$. Clearly this notion of a subsystem is
defined relative to a particular decomposition of $\LambdaM$ and
extrema of $S$ will not in general coincide with extrema of $\int_\MST
(\Lambda^{\Ein} + \Lambda^s)$.  When such extrema do coincide one may
argue that the subsystem described by $\Lambda^s$ decouples from the
system described by $\LambdaT$, i.e. it becomes a closed
subsystem. This is rarely the case in systems interacting with
dynamical gravitation. Subsystems that are not closed are termed open.

Einstein's gravitational field equations in the presence of matter
lead one to identify the variational derivative of $\LambdaM$ with
respect to $g$ (see below) with the symmetric
``stress-energy-momentum'' tensor $T^{\Met}(g, \{\boldsymbol Z\},
\{\boldsymbol\zeta\},\{d\boldsymbol\zeta\} )$ associated with
matter. This terminology is natural, given its historic connection
with Newtonian concepts but arguably misleading in a broader context
where symmetries associated with space and time translation are
absent.  The tensor $T^\Met$ relies for its definition on $\LambdaM$
depending on the spacetime metric.  On manifolds where gravitation is
irrelevant (where $g$ is regarded as a non-dynamical prescribed
background) or no preferred metric is available one may find
alternative approaches leading to conserved quantities. In such
circumstances there exist 3-forms derived from a diffeomorphism
invariant actions $\int_{\MST}\LambdaM$ that give rise to certain
vector valued maps on vector fields. These give rise to two other maps
$T^\Noet$ and $T^\LSym$ which have traditionally been associated with
``canonical stress-energy-momentum tensors'' and are related by the so
called Belinfante-Rosenfeld procedure. It is shown below that these
give rise to conserved currents in the presence of background fields
which possess appropriate Lie-symmetries and that $T^\Met$ gives rise
to conserved currents, in general, only if $g$ is dynamical and
satisfies Einstein's theory of general relativity. Furthermore it is
demonstrated that any difference between $T^\LSym$ and $T^\Met$
arises from the dependence of $\LambdaM$ on the background fields
$\Set{\boldsymbol Z}$.  In this article the consequences of
diffeomorphism invariance of matter actions on all these quantities
will be explored.

The formalism below first establishes
an intrinsic variational calculus for actions involving the metric
tensor field $g$, arbitrary tensor fields $\{\boldsymbol Z\} $ and
differential forms $ \Set{\boldsymbol\zeta}, \Set{d\boldsymbol\zeta} $ on an
$n$-dimensional manifold $M$.  Any differential form of degree $n$ on
$M$ will be called a top-form. The set of all $p$-form fields is
written $\Gamma\Lambda^p M$. Thus the matter action is some top-form:
\begin{align}
\LambdaM(g,Z_1,\ldots,Z_\NumZ,\zeta_1,d\zeta_1,\ldots,\zeta_\Numzeta,d\zeta_\Numzeta)
\in\Gamma\Lambda^n M
\label{Intro_Lambda}
\end{align}
and depends in general on
\begin{itemize}
\item
a  metric tensor field  $g$,

\item
a collection of mixed degree tensor fields (including scalar fields) of no particular symmetry  $Z_\Zind$ for $\Zind=1,\ldots,\NumZ$

\item
a collection of differential forms $\zeta_\zeind
\in\Gamma\Lambda^{p_\zeind} M $ for $\zeind=1,\ldots,\Numzeta$ and
their exterior derivatives $d\zeta_\zeind$. These fields will be taken to
satisfy variational field equations following from some action top-form
$\LambdaM$. Since $\Lambda^{\Ein}$ is independent of $ Z_\Zind $   and $ \zeta_\zeind $  these are the same as the variational equations that follow from the action top-form $\LambdaT$.

\end{itemize}

The notation used in the paper is given in section \ref{ch_App_Not}. In section \ref{ch_Tensors} partial Gateaux derivatives of $\LambdaM$ with respect to tensors and
differential forms are related and compared with standard partial
variational derivatives using an intrinsic
formulation. Diffeomorphism invariance of $\LambdaM$ is defined and
its consequences expressed in terms of certain maps ${\cal A}$ and
${\cal B}$ on these derivatives. In section \ref{ch_Phys} the Einstein-Hilbert and other
stress-energy-momentum tensors are defined and expressed in terms of
these maps. Relations are derived between these tensors and tensor
densities when all matter fields satisfy the variational field
equations derived from $\LambdaM$.  It is then shown how conserved
quantities can arise in the presence of material and Killing
Lie-symmetric background fields. In the concluding section the
physical implications of these relations are emphasized. The Appendix
gives some mathematical details and proofs of results used in the main text.

\subsection{Notation}
\label{ch_App_Not}

Local coordinates on $M$ are denoted $(x^1,\ldots,x^n)$. These define
a local coordinate frame $\Set{\partial_1,\ldots,\partial_n}$ where
$\partial_a=\pfrac{}{x^a}$ and coordinate co-frame
$\Set{dx^1,\ldots,dx^n}$. In these frames a metric $g=g_{ab} dx^a\otimes dx^b$
and the inverse metric $\dual{g}=g^{ab} \partial_a\otimes
\partial_b$ where $ g^{ab}g_{bc}= \delta^a_c $. Here implicit
summation is over $a,b=1,\ldots,n$.  The {\it metric dual} of any
vector field $v\in\Gamma TM$ is the 1-form
$\dual{v}=g(v,-)\in\Gamma\Lambda^1 M$. In spacetime $n=4$ and $g$ is
Lorentzian with signature $(-1,+1,+1,+1)$ here.  A map $T:X\mapsto T(X)$ is
$f$-linear if $T(fX)=f\,T(X)$ for all scalar fields
$f\in\Gamma\Lambda^0 M$. The map $T$ is $\Real$-linear if $T(\lambda
X)=\lambda\,T(X)$ for all constants $\lambda\in\Real$.  In the
following the word ``tensor'' refers to an $f$-multilinear map on
vectors and their duals (co-vectors).  Given a non-vanishing top-form
$\Omega$, the phrase ``$\Tden$ is a tensor density with respect to
$\Omega$ of weight $W$'' implies that when $\Omega$ is replaced by
$\hat\Omega=J \Omega$ for $J\in\Gamma\Lambda^0 M$ nowhere vanishing
then $\Tden$ transforms to ${\hat\Tden}=J^{W} \Tden$.  If
$\Omega=dx^1\wedge\cdots\wedge dx^n$ and
$\hat\Omega=dy^1\wedge\cdots\wedge dy^n$ are related by coordinate
transformation then $J$ is the Jacobian of the transformation. Choosing
$\Omega=\star 1$ one can convert $\Tden$ into a bona-fide tensor field
($\hat\Tden=\Tden$).  The Lie derivative of any tensor $T$ with
respect to any vector field $v$ on $M$ is denoted $ {\cal L}_v T $ and
the exterior derivative $d$ on differential forms is defined so that
$d^2=0$. A form $\beta$ is said to be closed if $d\beta=0$ and exact
if $\beta=d\alpha$ for some $\alpha$. The interior contraction
operator with respect to $v$ on forms, denoted $i_v$, is a graded
derivative and $(i_v)^2=0$.

\section{Intrinsic Variational Calculus}
\label{ch_Tensors}

\subsection{Algebraic Preliminaries}

The degree of an arbitrary tensor will be represented as an ordered
list $\slist$ of 0 or more entries. Each entry is either the symbol
$\Tform$ (for 1-form) or $\Tvec$ (for vector) e.g.
$\slist=[\Tform,\Tvec,\Tvec,\Tform]$. The bundle of tensors of degree
$\slist$ over $M$ is denoted $\tenbun^\slist M$, with sections in
$\Gamma\tenbun^\slist M$ and the bundles of 0-forms, 1-forms and
vector fields are written
\begin{align*}
\Lambda^0 M=\tenbun^{[\,]} M
\,,\qquad
\Lambda^1 M=\tenbun^{[\Tform]} M
\qquadand
TM=\tenbun^{[\Tvec]} M
\end{align*}
respectively.  Furthermore
\begin{align*}
(\tenbun^\slist M)\otimes(\tenbun^\tlist M)
=
\tenbun^{[\slist,\tlist]} M
\end{align*}
where $[\slist,\tlist]$ is simply the concatenation of the two lists.
Thus
\begin{align*}
v\otimes\zeta\otimes u\in\Gamma\tenbun^{[\Tvec,\Tform,\Tvec]} M
\end{align*}
where $v,u\in\Gamma TM$ and $\zeta\in\Gamma\Lambda^1 M$. The metric
$g$ lies in the symmetric sub-bundle of $ \tenbun^{[\Tform,\Tform]}M$
and the inverse metric $\tilde{g}$ lies in the symmetric sub-bundle of
$\tenbun^{[\Tvec,\Tvec]}M$.  Similarly, since
\begin{align}
\sum_{I_1<\ldots<I_p}\alpha_{I_1\cdots I_p}
e^{I_1}\wedge\cdots\wedge e^{I_p}
=
\frac1{p\,!}
\sum_{I_1<\ldots<I_p} \sum_{\sigma\in S_p}\alpha_{I_1\cdots I_p}
\epsilon(\sigma) e^{\sigma(I_1)}\otimes\cdots\otimes e^{\sigma(I_p)}
\label{Tensors_identify_form_tens}
\end{align}
where $S_p$ is the set of permutations of $\Set{I_1,\ldots,I_p}$
and $\epsilon(\sigma)$ is the signature of the permutation,  the
$p$-form bundle $\Lambda^p M$ is the antisymmetric sub-bundle of
$\tenbun^{[\Tform,\ldots,\Tform]}M$ where $ [\Tform,\ldots,\Tform] $
has length $p$.

The dual space of $\tenbun^\slist M$ is $\tenbun^{\cnj{\slist}} M$
where $\cnj{\slist}$ is obtained by interchanging the symbols $\Tform$
and $\Tvec$ in $\slist$.
The total contraction of elements in $\Gamma\tenbun^{\cnj{\slist}} M$ with
elements in $\Gamma\tenbun^\slist M$ is written
\begin{align*}
\Gamma\tenbun^{\cnj{\slist}} M \times \Gamma\tenbun^\slist M \to
\Gamma\Lambda^0 M
\,,\qquad
(\Phi,Z)\mapsto \Phi\Tprod Z
\end{align*}
where
$\Phi\in\Gamma\tenbun^{\cnj{\slist}} M$ and
$Z\in\Gamma\tenbun^{{\slist}} M$. It is defined inductively via
\begin{align*}
\zeta\Tprod v = v \Tprod \zeta = \zeta(v)
\qquadtext{where}
\zeta\in\Gamma\tenbun^{[\Tform]} M
\qquadand
v\in\Gamma\tenbun^{[\Tvec]} M
\end{align*}
and extended by $f$-linearity to arbitrary tensors via
\begin{align*}
(\Phi_1\otimes \Phi_2) \Tprod (Z_1\otimes Z_2)
=
(\Phi_1\Tprod Z_1)(\Phi_2 \Tprod Z_2)
\end{align*}
Thus for example
\begin{align*}
(\zeta_1\otimes u\otimes\zeta_2)\Tprod
(v_1\otimes\beta\otimes v_2)
=
\zeta_1(v_1)\,
\beta(u)\,
\zeta_2(v_2)
\end{align*}
where
$(\zeta_1\otimes u\otimes\zeta_2)\in\Gamma\tenbun^{[\Tform,\Tvec,\Tform]}M$
and
$(v_1\otimes\beta\otimes v_2)\in\Gamma\tenbun^{[\Tvec,\Tform,\Tvec]}M$.

Using this notation $\LambdaM$ may be regarded as a fibre bundle morphism
\begin{align}
\LambdaM:
\Ebun^{(g)}\oplus
\Ebun^{(Z)}\oplus
\Ebun^{(\zeta)}
\to
\Lambda^n M
\label{Tensors_Lambda}
\end{align}
which we write
$\LambdaM(g,Z_1,\ldots,Z_\NumZ,\zeta_1,d\zeta_1,\ldots,\zeta_\Numzeta,d\zeta_\Numzeta)$. Here
\begin{itemize}
\item
$\Ebun^{(g)} = \tenbun^{[\Tform,\Tform]} M$ is the bundle of metrics,
  i.e. $g\in\Gamma\Ebun^{(g)}$.
\item
$\Ebun^{(Z)} = \tenbun^{\slist_1} M \oplus \cdots
  \oplus\tenbun^{\slist_\NumZ} M$ is the bundle of $\NumZ$ tensors of
  the appropriate degrees, i.e
  $(Z_1,\ldots,Z_{\NumZ})\in\Gamma\Ebun^{(Z)}$ where
  $Z_\Zind\in\Gamma\tenbun^{\slist_\Zind} M$.
\item
$\Ebun^{(\zeta)} = \Lambda^{p_1} M \oplus \Lambda^{p_1+1} M \oplus
  \cdots \oplus \Lambda^{p_\Numzeta} M \oplus \Lambda^{p_\Numzeta+1}
  M$ is the bundle of $\Numzeta$ pairs of forms of the appropriate
  degrees, i.e
  $(\zeta_1,d\zeta_1,\ldots,\zeta_\Numzeta,d\zeta_\Numzeta)\in\Gamma\Ebun^{(\zeta)}$
  where $\zeta_\zeind\in\Gamma\Lambda^{p_\zeind} M$.
\end{itemize}
Since $\LambdaM$ is a fibre bundle morphism, the value of
$\LambdaM(g,Z_1,\ldots,Z_\NumZ,
\\\zeta_1,d\zeta_1,\ldots,\zeta_\Numzeta,d\zeta_\Numzeta)|_x$
for some point $x\in M$ depends only on the values of its arguments at
that point, i.e.
\begin{align*}
\lefteqn{
\LambdaM(g,Z_1,\ldots,Z_\NumZ,\zeta_1,d\zeta_1,\ldots,\zeta_\Numzeta,d\zeta_\Numzeta)|_x=}
\qquad&\\
&
\LambdaM(g|_x,Z_1|_x,\ldots,Z_\NumZ|_x,\zeta_1|_x,d\zeta_1|_x,\ldots,\zeta_\Numzeta|_x,d\zeta_\Numzeta|_x)
\end{align*}

In this language a model containing matter in a linear (temporally and spatially)
non-dispersive medium interacting with the electromagnetic field
$F=dA$ is characterised by a $U(1)$ gauge invariant excitation tensor
$G=\CR(g,Z_1,\ldots,Z_\NumZ,dA,\zeta_2,d\zeta_2,\ldots,\zeta_\Numzeta,
d\zeta_\Numzeta)$ for some constitutive tensor $\CR$ and is described by
an action $4$-form on spacetime $\MST$
\begin{equation}
\begin{aligned}
\lefteqn{
\LambdaM(g,Z_1,\ldots,Z_\NumZ,A,dA,\zeta_2,d\zeta_2,\ldots,
\zeta_\Numzeta,d\zeta_\Numzeta)
=}\qquad\qquad\qquad &
\\& \tfrac12 F\wedge\star G
+\Lambda^{\textbf{Q}}(g,\zeta_2,D\zeta_2,\ldots,
\zeta_\Numzeta,D\zeta_\Numzeta)
\end{aligned}
\label{Tensors_F_wedge_G}
\end{equation}
where $D\xi=d\xi$ for electrically neutral real fields $\xi$ and is the
$U(1)$ exterior covariant derivative for complex charged fields.

A particular model \cite{dereli2007covariant,dereli2007new} involving
only a single non-dynamic tensor $Z$ together with $F$ and $g$ is
described by the action 4-form $\LambdaM :
\Ebun^{(g)}\oplus\Ebun^{(Z)}\oplus\Ebun^{(\zeta)}\to\Lambda^4 \MST$
(with the bundles $\Ebun^{(Z)}=\tenbun^{[\Tform,\Tform,\Tvec,\Tvec]}
\MST$, $\Ebun^{(\zeta)}=\Lambda^2 \MST$ over spacetime $\MST$):
\begin{align}
&\LambdaM(g,Z,dA) = \tfrac14 F \wedge \star \big(Z(F)+Z^\dagger(F)\big)
\label{Intro_Lag_SymMink_1}
\end{align}
with
\begin{align}
F=dA
\label{Intro_F_dA}
\end{align}
and where $Z^\dagger$ is the adjoint of $Z$ defined by
\begin{align}
\alpha\wedge\star Z^\dagger(\beta)=
\beta\wedge\star Z(\alpha)
\label{Intro_def_adjoint}
\end{align}
for all $\alpha,\beta\in\Gamma\Lambda^2\MST$. Unlike $Z$, the tensor
$Z^\dagger$ depends on $g$.  Varying (\ref{Intro_Lag_SymMink_1}) with
respect to $A$ then yields, in terms of the notation defined in
(\ref{Tensors_general_V}) below
\begin{align*}
\frac{\delta \LambdaM}{\delta A} (g,Z,dA) = 0
\end{align*}
i.e.
\begin{align}
d\star G=0
\label{Intro_dstarG}
\end{align}
where in terms of (\ref{AB_def_pfrac_zeta}) below,
the excitation tensor $G\in\Gamma\Lambda^2\MST$ is given by
\begin{align}
G=\star^{-1} \Big(\pfrac{\LambdaM}{(dA)}\Big)
\label{Intro_def_G}
\end{align}
i.e
\begin{align}
G = \CR(F) = \tfrac12\big(Z(F)+Z^\dagger(F)\big)
\label{Intro_Lag_SymMink_2}
\end{align}
Equations (\ref{Intro_F_dA}) and (\ref{Intro_dstarG}) constitute the
``on-shell''
Maxwell system for model (\ref{Intro_Lag_SymMink_1}) in any
background  $g$, $Z$.
In this model the constitutive tensor
$\CR$ is independent of the motion of the medium. Although from
(\ref{Intro_def_adjoint}) $Z^\dagger$ depends on the metric it follows
from (\ref{Intro_Lag_SymMink_2}) that (\ref{Intro_Lag_SymMink_1}) can
be written
\begin{align}
\LambdaM(g,Z,dA) = \tfrac12 F \wedge\star Z(F)
\label{Intro_Lag_SymMink_3}
\end{align}
and thus the only metric dependence of $\LambdaM$ is through the Hodge map.

A more complex model \cite{dereli2007new} in which the constitutive
tensor $\CR$ depends explicitly on the motion of the medium and
exhibits intrinsic magneto-electric constitutive properties involves a
timelike\footnote{Since small variations in the metric $g$ do not
  change the timelike nature of $V$, its Gateaux derivative with
  respect to $g$ is zero.} vector field $V$ and four background degree 2
tensors $ Z^{\text{de}},Z^{\text{db}}, Z^{\text{he}},Z^{\text{hb}} $.
It is described by the action 4-form $\LambdaM :
\Ebun^{(g)}\oplus\Ebun^{(Z)}\oplus\Ebun^{(\zeta)}\to\Lambda^4 \MST$
where $\Ebun^{(Z)}=\tenbun^{[\Tform,\Tvec]} \MST
\oplus\tenbun^{[\Tform,\Tvec]} \MST \oplus\tenbun^{[\Tform,\Tvec]}
\MST \oplus\tenbun^{[\Tform,\Tvec]} \MST\oplus\tenbun^{[\Tvec]} \MST$,
$\Ebun^{(\zeta)}=\Lambda^2 \MST$,
\begin{equation}
\begin{aligned}
\lefteqn{\LambdaM(g,Z^{\text{de}},Z^{\text{db}},
Z^{\text{he}},Z^{\text{hb}},V,dA)}\qquad\qquad
\\
&=
\tfrac12 F \wedge \star\Big(
Z^{\text{de}}(i_{V_g} F)\wedge\dual{{V_g}} +
Z^{\text{db}}(i_{V_g}\star F)\wedge\dual{{V_g}}
\\&\qquad\qquad
- \star(Z^{\text{he}}(i_{V_g} F)\wedge\dual{{V_g}})
- \star(Z^{\text{hb}}(i_{V_g} \star F)\wedge\dual{{V_g}})
\Big)
\end{aligned}
\label{Intro_Lag_Abraham_1}
\end{equation}
where $F=dA$,
\begin{align}
V_g = \frac{V}{\sqrt{-g(V,V)}}
\,,\qquad
\dual{{V_g}}=g(V_g,-)
\label{Intro_def_Vg}
\end{align}
and the subscripts $g$ indicate explicit dependence on the metric.
Again variation with respect to $A$ gives the Maxwell equation
(\ref{Intro_dstarG}) where $G$ is given by (\ref{Intro_def_G}). Thus
\begin{equation}
\begin{aligned}
G
&=
Z^{\text{de}}_g(i_{V_g} F)\wedge\dual{{V_g}} +
Z^{\text{db}}_g(i_{V_g}\star F)\wedge\dual{{V_g}}
\\&\qquad\qquad
- \star(Z^{\text{he}}_g(i_{V_g} F)\wedge\dual{{V_g}})
- \star(Z^{\text{hb}}_g(i_{V_g} \star F)\wedge\dual{{V_g}})
\end{aligned}
\label{Intro_Lag_Abraham_2}
\end{equation}
where
\begin{equation}
\begin{gathered}
Z^{\text{de}}_{g} =
\tfrac12 \pi_g\circ \big(Z^{\text{de}} +
(Z^{\text{de}})^\dagger\big)\circ \pi_g
\,,\quad
Z^{\text{db}}_{g} =
\tfrac12 \pi_g\circ \big(Z^{\text{db}} -
(Z^{\text{he}})^\dagger\big)\circ \pi_g\,,
\\
Z^{\text{he}}_{g} =
\tfrac12 \pi_g\circ \big(Z^{\text{he}} -
(Z^{\text{db}})^\dagger\big)\circ \pi_g
\,,\quad
Z^{\text{hb}}_{g} =
\tfrac12 \pi_g\circ \big(Z^{\text{hb}} +
(Z^{\text{hb}})^\dagger\big)\circ \pi_g\,,
\\
\pi_g=\text{Id}_4 + \dual{V_g}\otimes V_g
\qquadand
\alpha\wedge\star Z_g^{\text{I}}(\beta)
=
\beta\wedge\star (Z_g^{\text{I}})^\dagger(\alpha)
\end{gathered}
\label{Intro_Lag_Abraham_Zg}
\end{equation}
for
$Z^{\text{I}}_g\in\Set{Z^{\text{de}}_g,
Z^{\text{db}}_g,Z^{\text{he}}_g,Z^{\text{hb}}_g}$
and $\alpha,\beta\in\Gamma\Lambda^1 M$. This implies $Z^{\text{I}}_g$
is spatial with respect to $V_g$, i.e. $Z^{\text{I}}_g(\dual{V_g})=0$
and $i_{V_g} Z^{\text{I}}_g(\alpha)=0$ for all $\alpha\in\Gamma\Lambda^1
M$. Since $\pi_g(\alpha)\wedge\dual{V_g}=\alpha\wedge\dual{V_g}$ and
$\pi_g (i_{V_g}\gamma)=i_{V_g}\gamma$ for all
$\alpha\in\Gamma\Lambda^1M$ and $\gamma\in\Gamma\Lambda^2M$ then
$Z^{\text{de}}_g$ may be replaced by $\tfrac12 \big(Z^{\text{de}} +
(Z^{\text{de}})^\dagger\big)$ in $G$. Similarly $Z^{\text{db}}_g$ may
be replaced by $\tfrac12 \big(Z^{\text{db}} -
(Z^{\text{he}})^\dagger\big)$, $Z^{\text{he}}_g$ by $\tfrac12
\big(Z^{\text{he}} - (Z^{\text{db}})^\dagger\big)$ and
$Z^{\text{hb}}_g$ by $\tfrac12 \big(Z^{\text{hb}} +
(Z^{\text{hb}})^\dagger\big)$.  This expresses $G$ more simply in
terms of the constitutive tensors in the action top form
(\ref{Intro_Lag_Abraham_1}). Furthermore, after some rearrangement,
one finds that $ \frac{1}{2} F \wedge \star G= \LambdaM $, (cf
\cite{dereli2007new}).

To facilitate the presentation below, it proves useful to relabel
tensors in the arguments of $\LambdaM$ as
\begin{align}
\Zalt_0 = g\,,\qquad
\Zalt_\Zind = Z_\Zind\,,\qquad
\Zalt_{\NumZ+2\zeind-1} = \zeta_\zeind \qquadand
\Zalt_{\NumZ+2\zeind} = d\zeta_\zeind
\label{Intro_Z_rename}
\end{align}
for $\Zind=1,\ldots,\NumZ$ and $\zeind=1,\ldots,\Numzeta$, so that
\begin{align}
\LambdaM(\Zalt_0,\ldots \Zalt_{\NumZ+2\Numzeta})\in\Gamma\Lambda^n M
\label{Intro_Lambda_Z_rename}
\end{align}
The range
$\ZZind=0,\ldots,\NumZ+2\Numzeta$ will be used to index the
$\Zalt_\ZZind$.

In the following tensors of the form
$\Psi=\Omega\otimes\Phi\in\Gamma(\Lambda^n M\otimes
\tenbun^{\cnj{\slist}} M)$ where $\Omega\in\Gamma\Lambda^n M$ and
$\Phi\in\Gamma\tenbun^{\cnj{\slist}}M$ arise naturally by
``differentiating'' $\LambdaM$ with respect to one of its arguments.
One may contract such a tensor with $Y\in\Gamma\tenbun^{{\slist}}M$ to
isolate $\Omega$
\begin{align*}
\Gamma (\Lambda^n M\otimes \tenbun^{\cnj{\slist}} M)
\times \Gamma \tenbun^{{\slist}} M
\to
\Gamma \Lambda^n M\,,\qquad
(\Psi,Y)\mapsto \Psi\TprodL Y
\end{align*}
according to the rule
\begin{align}
(\Omega\otimes\Phi)\TprodL Y =
 (\Phi\Tprod Y)\Omega
\label{Tensors_Omega_Phi_Y}
\end{align}
where $\Omega\in\Gamma\Lambda^n M$ and
$\Phi\in\Gamma\tenbun^{\cnj{\slist}} M$
and
$Y\in\Gamma\tenbun^{{\slist}} M$

\subsection{Variational Derivatives}
\label{sch_Var_deriv}

Using the indexing notation (\ref{Intro_Z_rename}) and
(\ref{Intro_Lambda_Z_rename}), the Gateaux derivative
$\displaystyle\GatZ{\Zalt_\ZZind}\in
\Gamma(\Lambda^nM\otimes\tenbun^{\cnj{\slist_\ZZind}}M)$ of $\LambdaM$
with respect to $\Zalt_\ZZind\in\Gamma\tenbun^{{\slist_\ZZind}}M$ is
defined so that
\begin{align}
\GatZ{\Zalt_\ZZind}\TprodL Y
=
\dfrac{}{\varepsilon}\Big|_{\varepsilon=0}
\LambdaM\big(\Zalt_0,\ldots,\Zalt_{\ZZind-1},\Zalt_{\ZZind}+\varepsilon Y,
\Zalt_{\ZZind+1},\ldots,\Zalt_{\NumZ+2\Numzeta}\big)
\label{AB_def_pfrac}
\end{align}
for all $Y\in\Gamma\tenbun^{{\slist_\ZZind}}M$.
An example in a local frame is given in appendix \ref{sch_Examples}.

By contrast, for $\ZZind=\NumZ+2\zeind-1$ then
$\Zalt_{\NumZ+2\zeind-1}=\zeta_\zeind\in\Gamma\Lambda^{p_\zeind} M$ i.e.
$\zeta_\zeind\in\Gamma\tenbun^{[\Tform,\ldots,\Tform]}$ (with a list of
length $p_\zeind$) the  Gateaux derivative
$\displaystyle\pfrac{\LambdaM}{\zeta_\zeind}\in\Gamma\Lambda^{n-{p_\zeind}}M$ is
defined so that
\begin{equation}
\begin{aligned}
\alpha\wedge\pfrac{\LambdaM}{\zeta_\zeind} 
&=
\dfrac{}{\varepsilon}\Big|_{\varepsilon=0}
\LambdaM\big(g,Z_1,\ldots,Z_\NumZ,\zeta_1,d\zeta_1,
\ldots,\zeta_\zeind+\varepsilon\alpha,d\zeta_\zeind,\ldots,
\zeta_\Numzeta,d\zeta_\Numzeta\big)
\\
&= \GatZ{\zeta_\zeind}\TprodL \alpha
\end{aligned}
\label{AB_def_pfrac_zeta}
\end{equation}
for all $\alpha\in\Gamma\Lambda^{p_\zeind} M$. The derivative
$\displaystyle\pfrac{\LambdaM}{\zeta_\zeind}\in\Gamma\Lambda^{n-p_\zeind} M$ may be related  to the derivative
$\displaystyle\GatZ{\zeta_\zeind}\in\Gamma (\Lambda^n M\otimes
\tenbun^{[\Tvec,\ldots,\Tvec]})$ since one can identify
$\Gamma\Lambda^{n-p_\zeind} M$ and
$\big\{\Omega\otimes\Vasym \in\Gamma (\Lambda^n M\otimes
\tenbun^{[\Tvec,\ldots,\Tvec]})\,\big| \Vasym\text{ antisymmetric} \big\}$.
The correspondence follows from the relation:
\begin{align}
(\Omega\otimes \Vasym)\TprodL \alpha
=
\alpha\wedge i_\Vasym \Omega
\label{AB_Omega_X_alpha}
\end{align}
 for any $\alpha\in\Gamma\Lambda^{p}M$.
Here a general antisymmetric
tensor $\Vasym\in\Gamma\tenbun^{[\Tvec,\ldots,\Tvec]}M$ can be written
\begin{equation}
\begin{aligned}
\Vasym
&=
\sum_{I_1<\ldots<I_{p}} \Vasym^{I_1\cdots I_{p}}  X_{I_1}\wedge\cdots\wedge X_{I_{p}}
\\&=
\frac1{{p}!}
\sum_{I_1<\ldots<I_{p}} \Vasym^{I_1\cdots I_{p}}    \sum_{\sigma\in S_{p}}
\epsilon(\sigma) X_{\sigma(I_1)}\otimes\cdots\otimes X_{\sigma(I_{p})}
\end{aligned}
\label{Tensors_general_V}
\end{equation}
and the internal contraction operator $ i_\Vasym  $ with respect to
$\Vasym$ is defined by:
\begin{align}
i_\Vasym= \frac1{{p}!}\sum_{I_1<\ldots<I_{p}} \Vasym^{I_1\cdots I_{p}}\,
i_{X_{I_{p}}}\cdots i_{X_{I_1}}
\label{Tensors_def_icalV}
\end{align}
so that for $\alpha\in\Gamma\Lambda^{p} M$
\begin{align}
\alpha\Tprod \Vasym =
i_\Vasym \alpha
\label{Tensors_alpha_V}
\end{align}
The proofs of (\ref{AB_Omega_X_alpha}) and (\ref{Tensors_alpha_V}) are
given in lemmas \ref{lm_LagForm} and \ref{lm_alpha_V} respectively in
appendix \ref{sch_lammas}.

Likewise if $\ZZind=\NumZ+2\zeind$ then
$\Zalt_{\NumZ+2\zeind}=d\zeta_\zeind\in\Gamma\Lambda^{p_\zeind+1} M$ and
$\displaystyle\pfrac{\LambdaM}{(d\zeta_\zeind)}\in\Gamma\Lambda^{n-p_\zeind-1}M$ is
defined by
\begin{equation}
\begin{aligned}
\beta\wedge\pfrac{\LambdaM}{(d\zeta_\zeind)}
&=
\dfrac{}{\varepsilon}\Big|_{\varepsilon=0}
\LambdaM(g,Z_1,\ldots,Z_\NumZ,\zeta_1,d\zeta_1,\ldots,
\zeta_\zeind,d\zeta_\zeind+\varepsilon\beta,\ldots,\zeta_\Numzeta,d\zeta_\Numzeta)
\\&=
\GatZ{(d\zeta_\zeind)}\TprodL\beta
\end{aligned}
\label{LagForm_def_pfrac_d_alpha_mu}
\end{equation}
for all $\beta\in\Gamma\Lambda^{p_\zeind+1} M$.

The variational derivative
$\displaystyle\frac{\delta\LambdaM}{\delta\zeta_\zeind}$ of $\LambdaM$
with respect to $\zeta_\zeind\in\Gamma\Lambda^{p_\zeind} M$ is defined
by
\begin{equation}
\begin{aligned}
\int_M\beta\wedge\frac{\delta\LambdaM}{\delta\zeta_\zeind}
=
\dfrac{}{\varepsilon}\Big|_{\varepsilon=0}
\int_M
\LambdaM(&g,Z_1,\ldots,Z_\NumZ,\zeta_1,d\zeta_1,\ldots,
\\\quad&
\zeta_\zeind+\varepsilon\beta,d\zeta_\zeind+\varepsilon d\beta,
\ldots,\zeta_\Numzeta,d\zeta_\Numzeta)
\end{aligned}
\label{LagForm_dynamical_def}
\end{equation}
for all $\beta\in\Gamma\Lambda^{p_\zeind} M$ with {\it compact support}. Hence
\begin{align}
\frac{\delta\LambdaM}{\delta\zeta_\zeind}
=
\pfrac{\LambdaM}{\zeta_\zeind} + (-1)^{p_\zeind+1}
d\Big(\pfrac{\LambdaM}{(d\zeta_\zeind)}\Big)
\label{LagForm_def_delta_frac_ea}
\end{align}
follows from lemma
\ref{lm_variation} in appendix \ref{sch_lammas}.

A $p$-form  $\zeta_\zeind$ is  said to be {``on $\LambdaM$-shell''} if
\begin{align}
\frac{\delta\LambdaM}{\delta\zeta_\zeind} = 0
\label{LagForm_dynamical}
\end{align}
In this situation
\begin{align}
\pfrac{\LambdaM}{\zeta_\zeind} = (-1)^{p_\zeind}
d\Big(\pfrac{\LambdaM}{(d\zeta_\zeind)}\Big)
\label{LagForm_delta_frac_ea_res}
\end{align}
The matter system is said to be on $\LambdaM$-shell if
(\ref{LagForm_dynamical}) is true for all $\zeta_\zeind$.

\subsection{Diffeomorphism invariance}

Given a local diffeomorphism $\phi:U_M\to U_M'$ where $U_M,U_M'\subset
M$ then the two maps $\phi_\star$ and $\phi^{-1\star}$ induce the map
$\hat\phi:\tenbun^{{\slist}}U_M\to\tenbun^{{\slist}} U_M'$ on
tensors. The action top-form $\LambdaM(\Zalt_0,\ldots
\Zalt_{\NumZ+2\Numzeta})$ is {\it (locally) diffeomorphism invariant}
if
\begin{align}
\hat\phi\Big(\LambdaM(\Zalt_0,\Zalt_1,\ldots,\Zalt_{\NumZ+2\Numzeta})\Big)
=
\LambdaM\Big(\hat\phi(\Zalt_0),\hat\phi(\Zalt_1),\ldots,
\hat\phi(\Zalt_{\NumZ+2\Numzeta})\Big)
\label{AB_def_diffo}
\end{align}
for all local diffeomorphism $\phi:U_M\to U_M'$.
\begin{lemma}
If $\LambdaM$ is diffeomorphism invariant then for all vector fields
$v\in\Gamma TM$
\begin{align}
\Lie_v \big(\LambdaM(\Zalt_0,\Zalt_1,\ldots,\Zalt_{\NumZ+2\Numzeta})\big)
=
\sum_{\ZZind=0}^{\NumZ+2\Numzeta}
\GatZ{\Zalt_\ZZind}\TprodL \Lie_v \Zalt_\ZZind
\label{Intro_diffeo_Lie}
\end{align}
which may be written
\begin{equation}
\begin{aligned}
&\Lie_v
\big(\LambdaM(g,Z_1,\ldots,Z_\NumZ,\zeta_1,d\zeta_1,
\ldots,\zeta_\Numzeta,d\zeta_\Numzeta)\big)=
\\&
\GatZ{g}\TprodL \Lie_v g
+
\sum_{\Zind=1}^{\NumZ}
\GatZ{Z_\Zind}\TprodL \Lie_v Z_\Zind
+
\sum_{\zeind=1}^{\Numzeta}
\Lie_v \zeta_\zeind\wedge\pfrac{\LambdaM}{\zeta_\zeind}
+
\sum_{\zeind=1}^{\Numzeta}
\Lie_v (d\zeta_\zeind)\wedge\pfrac{\LambdaM}{(d\zeta_\zeind)}
\end{aligned}
\label{Intro_diffeo_Lie_forms}
\end{equation}

\end{lemma}

\begin{proof}
Let $\phi_\varepsilon$ be the one parameter family of diffeomorphisms
generated by $v$. From (\ref{AB_def_diffo})
\begin{align*}
\Lie_v \LambdaM(\Zalt_0,\ldots,\Zalt_{\NumZ+2\Numzeta})
&=
\dfrac{}{\varepsilon}\Big|_0
\hat\phi_\varepsilon\Big(\LambdaM(\Zalt_0,\ldots,\Zalt_{\NumZ+2\Numzeta})\Big)
\\&=
\dfrac{}{\varepsilon}\Big|_0
\LambdaM\Big(\hat\phi_\varepsilon(\Zalt_0),\ldots,\hat\phi_\varepsilon(\Zalt_{\NumZ+2\Numzeta})\Big)
\\&=
\sum_{\ZZind=0}^{\NumZ+2\Numzeta}
\dfrac{}{\varepsilon}\Big|_0
\LambdaM\Big(\Zalt_0,\ldots,\hat\phi_\varepsilon(\Zalt_\ZZind),\ldots,\Zalt_{\NumZ+2\Numzeta}\Big)
\\&=
\sum_{\ZZind=0}^{\NumZ+2\Numzeta}
\dfrac{}{\varepsilon}\Big|_0
\LambdaM\Big(\Zalt_0,\ldots,\Zalt_\ZZind + \varepsilon\Lie_v \Zalt_\ZZind,\ldots,\Zalt_{\NumZ+2\Numzeta}\Big)
\\&=
\sum_{\ZZind=0}^{\NumZ+2\Numzeta}
\GatZ{\Zalt_\ZZind}\TprodL \Lie_v \Zalt_\ZZind
\end{align*}
Then (\ref{Intro_diffeo_Lie_forms}) follows
from (\ref{AB_def_pfrac_zeta}) and
(\ref{LagForm_def_pfrac_d_alpha_mu}).
\end{proof}

Since $\Lie_v\zeta_\zeind=i_v d\zeta_\zeind + d i_v \zeta_\zeind$, the
top-forms $\displaystyle\Lie_v
\zeta_\zeind\wedge\pfrac{\LambdaM}{\zeta_\zeind}$ and $\displaystyle
\Lie_v(d \zeta_\zeind)\wedge\pfrac{\LambdaM}{(d\zeta_\zeind)}$ in
(\ref{Intro_diffeo_Lie}) can be expressed as
\begin{align}
\Lie_v\zeta_\zeind\wedge\pfrac{\LambdaM}{\zeta_\zeind}
=
d\Big(i_v \zeta_\zeind \wedge \pfrac{\LambdaM}{\zeta_\zeind}  \Big)
+
(-1)^{p_\zeind} i_v \zeta_\zeind\wedge d\pfrac{\LambdaM}{\zeta_\zeind} +
i_v d\zeta_\zeind\wedge\pfrac{\LambdaM}{\zeta_\zeind}
\label{AB_Lie_v_zeta}
\end{align}
and
\begin{align}
\Lie_v(d\zeta_\zeind)\wedge\pfrac{\LambdaM}{(d\zeta_\zeind)}
=
d\Big(i_v d\zeta_\zeind \wedge \pfrac{\LambdaM}{(d\zeta_\zeind)} \Big)
-
(-1)^{p_\zeind} i_v d\zeta_\zeind\wedge d\pfrac{\LambdaM}{(d\zeta_\zeind)}
\label{AB_Lie_v_dzeta}
\end{align}
where $\zeta_\zeind\in\Gamma\Lambda^{p_\zeind} M$. See lemma
\ref{lm_Lie_v_beta_alpha} in the appendix \ref{sch_lammas}.

By analogy with (\ref{AB_Lie_v_zeta}) and (\ref{AB_Lie_v_dzeta}), the
top-forms $\displaystyle \GatZ{\Zalt_\ZZind}\TprodL \Lie_v
\Zalt_\ZZind$ in (\ref{Intro_diffeo_Lie}) may be written
\begin{align}
\GatZ{\Zalt_\ZZind}\TprodL \Lie_v \Zalt_\ZZind
=
d\bigg(\Amap_v\Big(\GatZ{\Zalt_\ZZind},\Zalt_\ZZind\Big)\bigg)
+
\Bmap_v\Big(\GatZ{\Zalt_\ZZind},\Zalt_\ZZind\Big)
\label{AB_AB_GatZ}
\end{align}
where the maps
\begin{align*}
&\Amap:\Gamma TM \times
\Gamma (\Lambda^n M\otimes \tenbun^{\cnj{\slist}}M)
\times \Gamma\tenbun^{{\slist}}M
\to
\Gamma \Lambda^{n-1} M
\,,\quad
(v,\Psi,Z)\mapsto \Amap_v(\Psi,Z)
\end{align*}
and
\begin{align*}
&\Bmap:\Gamma TM \times
\Gamma (\Lambda^n M\otimes \tenbun^{\cnj{\slist}}M)
\times \Gamma\tenbun^{{\slist}}M
\to
\Gamma \Lambda^{n} M
\,,\quad
(v,\Psi,Z)\mapsto \Bmap_v(\Psi,Z)
\end{align*}
are `f'-linear in $v$ and defined inductively as follows:
For 0-forms, $f\in\Gamma\tenbun^{{[\,]}}M$
\begin{align}
\Amap_v(\Omega,f) = 0
\qquadand
\Bmap_v(\Omega,f) = v(f) \Omega
\label{AB_def_AB_f}
\end{align}
For 1-forms $\alpha\in\Gamma\tenbun^{{[\Tform]}}M$
\begin{equation}
\begin{aligned}
&\Amap_v(\Omega\otimes u,\alpha) = \alpha(v) i_u\Omega
\\\text{and}\qquad
&\Bmap_v(\Omega\otimes u,\alpha) =
(-1)^{n+1} i_u \Omega \wedge i_v d\alpha -
d i_u \Omega \wedge i_v \alpha
\end{aligned}
\label{AB_def_AB_alpha}
\end{equation}
For vectors $u\in\Gamma\tenbun^{[\Tvec]}M$
\begin{equation}
\begin{aligned}
&\Amap_v(\Omega\otimes \alpha,u) = -\alpha(v) i_u\Omega
\\\text{and}\qquad&
\Bmap_v(\Omega\otimes \alpha,u) =
v\big(\alpha(u)\big)\,\Omega
- i_u \Omega \wedge i_v d\alpha
+ (-1)^n d i_u \Omega \wedge i_v \alpha
\end{aligned}
\label{AB_def_AB_u}
\end{equation}
For tensors $\Phi_1\in\Gamma\tenbun^{\cnj{\slist}} M$,
$\Phi_2\in\Gamma\tenbun^{\cnj\tlist} M$, $Z_1\in\Gamma\tenbun^{\slist}
M$, $Z_2\in\Gamma\tenbun^{\tlist} M$ and $\Omega\in\Gamma\Lambda^n M$,
define the Leibnitz rules:
\begin{equation}
\begin{aligned}
&\Amap_v(\Omega\otimes \Phi_1\otimes\Phi_2,Z_1\otimes Z_2)
\\
&\hspace{4em}=
\Amap_v((\Phi_1\Tprod Z_1)\Omega\otimes \Phi_2,Z_2) +
\Amap_v((\Phi_2\Tprod Z_2)\Omega\otimes \Phi_1,Z_1)
\\
&\Bmap_v(\Omega\otimes \Phi_1\otimes\Phi_2,Z_1\otimes Z_2)
\\
&\hspace{4em}=
\Bmap_v((\Phi_1\Tprod Z_1)\Omega\otimes \Phi_2,Z_2) +
\Bmap_v((\Phi_2\Tprod Z_2)\Omega\otimes \Phi_1,Z_1)
\end{aligned}
\label{AB_def_AB_Ten}
\end{equation}
The proof that (\ref{AB_AB_GatZ}) follows from lemma \ref{lm_LvZ_dA_B}
given in appendix \ref{sch_lammas}. The proof that
(\ref{AB_Lie_v_zeta}) and (\ref{AB_Lie_v_dzeta}) are consistent with
(\ref{AB_AB_GatZ}) is given in lemma \ref{lm_Lie_v_beta_AB} appendix
\ref{sch_lammas}.

If $\LambdaM$ is diffeomorphism
invariant then from (\ref{Intro_diffeo_Lie}) and (\ref{AB_AB_GatZ}):
\begin{equation}
\begin{aligned}
d i_v \LambdaM
&=
\Lie_v \LambdaM
=
\sum_{\ZZind=0}^{\NumZ+2\Numzeta}
\GatZ{\Zalt_\ZZind}\Tprod \Lie_v \Zalt_\ZZind
\\&=
d\bigg(
\sum_{\ZZind=0}^{\NumZ+2\Numzeta}
\Amap_v \Big(\GatZ{\Zalt_\ZZind},\Zalt_\ZZind\Big)
\bigg)
+
\sum_{\ZZind=0}^{\NumZ+2\Numzeta}
\Bmap_v \Big(\GatZ{\Zalt_\ZZind},\Zalt_\ZZind\Big)
\end{aligned}
\label{AB_diVLam}
\end{equation}
hence applying lemma \ref{lm_AB}, appendix \ref{sch_lammas} gives
\begin{align}
i_v\LambdaM =
\sum_{\ZZind=0}^{\NumZ+2\Numzeta}
\Amap_v \Big(\GatZ{\Zalt_\ZZind},\Zalt_\ZZind\Big)
\label{AB_ivLambda_A}
\end{align}
and
\begin{align}
\sum_{\ZZind=0}^{\NumZ+2\Numzeta}\Bmap_v \Big(\GatZ{\Zalt_\ZZind},\Zalt_\ZZind\Big) = 0
\label{AB_ivLambda_B}
\end{align}
Using (\ref{AB_Lie_v_zeta}) and (\ref{AB_Lie_v_dzeta}), these may be written as
\begin{equation}
\begin{aligned}
\Amap_v\Big(\GatZ{g},g\Big)
&=
i_v\LambdaM
-
\sum_{\Zind=1}^\NumZ
\Amap_v \Big(\GatZ{Z_\Zind },Z_\Zind \Big)
\\&\qquad\qquad
-
\sum_{\zeind=1}^\Numzeta\Big(
i_v \zeta_\zeind\wedge \pfrac{\LambdaM}{\zeta_\zeind}
+
i_v d\zeta_\zeind\wedge \pfrac{\LambdaM}{(d\zeta_\zeind)} \Big)
\end{aligned}
\label{AB_tau_Mat}
\end{equation}
and
\begin{equation}
\begin{aligned}
\Bmap_v\Big(\GatZ{g},g\Big)
&=
-\sum_{\Zind=1}^\NumZ\Bmap_v \Big(\GatZ{Z_\Zind },Z_\Zind \Big)
-\sum_{\zeind=1}^\Numzeta\bigg(
(-1)^{p_\zeind} i_v \zeta_\zeind\wedge d\pfrac{\LambdaM}{\zeta_\zeind}
\\&\qquad\qquad
 + i_v d\zeta_\zeind\wedge\pfrac{\LambdaM}{\zeta_\zeind} +
(-1)^{p_\zeind+1} i_v d\zeta_\zeind\wedge d\pfrac{\LambdaM}{(d\zeta_\zeind)}
\bigg)
\end{aligned}
\label{AB_Bg_D_tau_Met}
\end{equation}
When the system is on $\LambdaM$-shell (\ref{AB_tau_Mat}) may be written
\begin{equation}
\begin{aligned}
\Amap_v\Big(\GatZ{g},g\Big)
&=
i_v\LambdaM
-
\sum_{\zeind=1}^\Numzeta
\Lie_v \zeta_\zeind \wedge \pfrac{\LambdaM}{(d\zeta_\zeind)}
\\&
\qquad+
d\Big(\sum_{\zeind=1}^\Numzeta
i_v \zeta_\zeind\wedge \pfrac{\LambdaM}{(d\zeta_\zeind)}\Big)
-
\sum_{\Zind=1}^\NumZ
\Amap_v \Big(\GatZ{Z_\Zind },Z_\Zind \Big)
\end{aligned}
\label{AB_tau_Mat_Shell}
\end{equation}
whereas (\ref{AB_Bg_D_tau_Met}) reduces to
\begin{align}
\Bmap_v\Big(\GatZ{g},g\Big)
&=
-\sum_{\Zind=1}^\NumZ\Bmap_v \Big(\GatZ{Z_\Zind },Z_\Zind \Big)
\label{AB_Bg_D_tau_Met_Shell}
\end{align}
See lemma \ref{lm_Lie} appendix \ref{sch_lammas}.  Relations
(\ref{AB_tau_Mat})-(\ref{AB_Bg_D_tau_Met_Shell}) play a pivotal role
in the arguments below. In particular they enable one identify terms
which may be associated with the quantities derived historically with
Noether, Belinfante and Rosenfeld.

\section{Currents and Conservation laws}
\label{ch_Phys}

Many quantities in physics owe their raison-d'etre to the existence of conserved quantities that do not change with time in a dynamical system. Thus notions of energy, momentum and angular momentum arose from the analysis of Newtonian {\it particle} dynamics. With the introduction of {\it fields} and the development of continuum mechanics it became natural to incorporate such concepts into continuous dynamical systems and their unification into a ``stress-energy-momentum'' complex offered an attractive objective. However, as is well known such a unification is not unique and the indiscriminate use of the term ``stress-energy-momentum'' tensor has led to unnecessary confusion when discussing forces and torques produced by fields in media.

In this section a number of technical issues are addressed that inter-relate these physical concepts. They include the role played by different aspects of (multi-)linearity needed for a general definition of stress over curved surfaces in space, the role of Killing symmetry needed to establish the notions of energy and momentum and the role of algebraic symmetry of maps and associated tensors or tensor-densities in their conservation. Using the variational  framework established in the previous section it is possible to correlate these aspects with the parts played by the presence or absence of background matter fields and background gravitation in their implementation.

Since some of the traditional arguments for the construction of a symmetric  stress-energy-momentum tensor in Minkowski spacetime are spurious (even in the absence of background matter fields) it is useful to begin  the discussion with the Noether current associated with  $\LambdaM$ in a general background using the Lie-derivative. This leads naturally to conservation laws in the presence of background symmetries.  In this manner it is also straightforward to extricate the role played by $f$-linearity in establishing a proper tensor description of stress.

\subsection{Noether and Belinfante-Rosenfeld currents}

For the restricted class of actions considered in this article we
define the Noether $(n-1)$-form current associated with $\LambdaM$ by
\begin{align}
\Noet_v =
i_v \LambdaM -
\sum_{\zeind=1}^\Numzeta
\Lie_v \zeta_\zeind \wedge \pfrac{\LambdaM}{(d\zeta_\zeind)}
\label{LagForm_def_tau_can}
\end{align}
for $v\in\Gamma TM$, where it is assumed that the system is on
$\LambdaM$-shell. Since the Lie derivative $\Lie_v$ is not $f$-linear
in $v$ neither is $\Noet_v$ and for $v=v^a \partial_a$, lemma
\ref{lm_L_fv_alpha} in appendix \ref{sch_lammas} gives
\begin{align}
\Noet_v = v^a \Noet_{\partial_a} + \sum_{\zeind=1}^\Numzeta
d v^a \wedge i_{\partial_a}  \zeta_\zeind \wedge \pfrac{\LambdaM}{(d\zeta_\zeind)}
\label{LagForm_tau_can_lin}
\end{align}
It is however $\Real$-linear in $v$ (where $v^a$ are constants).  For
a chosen nowhere vanishing $\Omega\in\Gamma\Lambda^nM$ one may define
the map $\Tden^\Noet:\Gamma TM\to\Gamma TM$ by
\begin{align}
\Noet_v = i_{\Tden^\Noet(v)}\Omega
\label{Tensorsdef_T_Noet}
\end{align}
In a coordinate system $(x^1,\ldots,x^n)$ with
$\Omega=dx^1\wedge\cdots\wedge dx^n$ and $\LambdaM={\Sden}\Omega$ then
$\Tden^\Noet$ has component maps
\begin{align}
\Tden^\Noet{}^a{}_b = dx^a\big(\Tden^\Noet(\partial_b)\big)
\label{Tensors_T_Neot_coords_def}
\end{align}
which from lemma \ref{lm_tau_Tden} in appendix \ref{sch_lammas} gives
\begin{align}
dx^a \wedge \Noet_{\partial_b} = \Tden^\Noet{}^a{}_b \,\Omega
\qquadand
\Tden^\Noet{}^a{}_b =
i_{\partial_n} \cdots i_{\partial_1} (dx^a \wedge \Noet_{\partial_b})
\label{Tensors_T_Neot_coords_alt}
\end{align}
Lemma
\ref{lm_coord_T} in appendix \ref{sch_lammas} then gives
\begin{align}
\Tden^\Noet{}^a{}_b
=
\delta_b^a {\Sden}
-
\sum_{\zeind=1}^\Numzeta \quad
\sum_{I_1<\cdots<I_{p_B}}
\partial_b(\zeta_{BI})
\pfrac{{\Sden}}{(\partial_a\zeta_{BI})}
\label{Tensors_T_Neot_coords}
\end{align}
where
\begin{align*}
\zeta_B=\sum_{I_1<\cdots<I_{p_B}}
\zeta_{BI} \,dx^{I_1}\wedge\cdots\wedge
dx^{I_{p_B}}
\end{align*}
Since $\Tden^\Noet$ is not $f$-linear, i.e. $\Tden^\Noet(fv)\ne
f\Tden^\Noet(v)$ in general, one must not confuse
$\Tden^\Noet{}_b{}^a$ with the components of a tensor field.
In fact using (\ref{LagForm_tau_can_lin}) and lemma \ref{lm_tau_Tden}
in the appendix for $v=v^a\partial_a$,
\begin{align}
\Tden^\Noet(v)
=
(\Tden^\Noet)^a{}_b v^b
+
(\partial_c v^b) i_{\partial_1} \cdots i_{\partial_n}\Big(
\sum_{\zeind=1}^\Numzeta
dx^a\wedge
d x^c \wedge i_{\partial_b}  \zeta_\zeind \wedge \pfrac{\LambdaM}{(d\zeta_\zeind)}
\Big)
\label{Tensors_TNv_ab}
\end{align}
However, it is not uncommon to refer to (\ref{Tensors_T_Neot_coords})
as the components of the canonical stress-energy-momentum tensor
associated with $\LambdaM$ in
Minkowski spacetime
 \cite{goldstein,soper1976classical,wald1984general}.
This arises since $\Tden^\Noet{}_b{}^a\to R^a{}_c \Tden^\Noet{}_b{}^c$,
under affine coordinate transformations of the form
\begin{align*}
x^a \to y^a = R^a{}_b x^b + B^a
\end{align*}
where $\det(R)=1$ and $R^a{}_b,B^a\in\Real$ are constants.  This
accounts for its widespread use in special relativity as formulated by
Einstein and Minkowski in spacetime $\MST$.  Since $\Tden^\Noet$ is
not tensorial with respect to arbitrary coordinate transformations its
use for calculating stresses is restricted to planar surfaces in
space. This follows from (\ref{Tensors_TNv_ab}),
since for any event $p\in\MST$ on a {\em non-planar} spacelike
2-surface with normal field $w$, $\Tden^\Noet(w)|_p$ will depend on
the derivatives $(\partial_a w^b)|_p$. Thus one cannot, in general,
define Cauchy traction forces that must be independent of such
derivatives \cite{tucker2008intrinsic}.

The requirement that the concept of stress follows from a bona-fide
tensor leads one to seek an $f$-linear map constructed from
$\Noet_v$. Since $\Lie_v= d i_v + i_v d$, the $(n-1)$-form current
\begin{align}
\tau^\LSym_v
=
\Noet_v + dS_v
\label{Phys_tau_Lsym_can}
\end{align}
where
\begin{align}
S_v = \sum_{\zeind=1}^\Numzeta
i_v \zeta_\zeind\wedge \pfrac{\LambdaM}{(d\zeta_\zeind)}
\label{Phys_def_Sv}
\end{align}
is manifestly $f$-linear in $v$.  One may refer to this as a
Belinfante-Rosenfeld formula  \cite{gotay1992stress} although one may
note that no metric on $M$ is necessary for its construction. It follows from
(\ref{LagForm_def_tau_can}) and (\ref{Phys_def_Sv}) that
\begin{align}
\tau^\LSym_v
=
i_v\LambdaM
-
\sum_{\zeind=1}^\Numzeta\Big(
i_v \zeta_\zeind\wedge \pfrac{\LambdaM}{\zeta_\zeind}
+
i_v d\zeta_\zeind\wedge \pfrac{\LambdaM}{(d\zeta_\zeind)} \Big)
\label{Phys_tau_Lsym_alt}
\end{align}
which we refer to as the {\em Belinfante-Rosenfeld} stress-energy-momentum
current associated with $\LambdaM$.
From (\ref{AB_tau_Mat}) one can also write
(\ref{Phys_tau_Lsym_alt}) in terms of the maps $\Amap_v$ as
\begin{align}
\tau^\LSym_v
=
\Amap_v\Big(\GatZ{g},g\Big)
+
\sum_{\Zind=1}^\NumZ
\Amap_v \Big(\GatZ{Z_\Zind },Z_\Zind \Big)
\label{Phys_tau_Lsym_A}
\end{align}
Again, for a chosen nowhere vanishing
$\Omega\in\Gamma\Lambda^nM$ the map
$\Tden^\LSym\in\Gamma\tenbun^{[\Tvec,\Tform]}$ defined\footnote{Some authors \cite{soper1976classical,wald1984general} refer to
  $\Tden^\Noet$ as the canonical stress-energy-momentum tensor whereas
  others \cite{hehl2003foundations,obukhov2008electromagnetic,hehl1995metric}
  refer to $\Tden^\LSym$ as the canonical stress-energy-momentum
  tensor. The appellations $\Tden^\Noet$ and $\Tden^\LSym$ eliminate this
  notational ambiguity.}  by
\begin{align}
\tau^\LSym_v = i_{\Tden^\LSym(v)}\Omega
\label{Tensorsdef_T_LSym}
\end{align}
is a density with respect to $\Omega$ with weight $-1$.

Given a preferred metric $g$ and any map $\Tden:\Gamma TM\to\Gamma TM$
one may define the map
$T:\Gamma TM\times\Gamma TM\to \Gamma\Lambda^0 M$ by
\begin{align*}
T(u,v)=g(\Tden(u),v)
\end{align*}
This enables one to discuss the algebraic symmetries of $T$.  Such a
map is said to be {\em algebraically symmetric} with respect to $g$ if
\begin{align}
T(u,v)=T(v,u)
\label{Tensors_symm}
\end{align}
This implies (see lemma \ref{lm_T_tau_symmetry} in appendix \ref{ch_Id})
\begin{align}
\dual{v}\wedge i_{\Tden(u)} \Omega
=
\dual{u}\wedge i_{\Tden(v)} \Omega
\label{Tensors_symm_alt}
\end{align}
for all $u,v\in\Gamma TM$ any non-vanishing top-form $\Omega\in\Lambda^n M$.
Thus for the maps
$T^\Noet:\Gamma TM\times\Gamma TM\to \Gamma\Lambda^0 M$ and
$T^\LSym:\Gamma TM\times\Gamma TM\to \Gamma\Lambda^0 M$
\begin{align}
T^\Noet(u,v)=g\big(\Tden^\Noet(u),v\big)
\qquadand
T^\LSym(u,v)=g\big(\Tden^\LSym(u),v\big)
\label{Tensors_TN_and_TB}
\end{align}
Then $T^\Noet$ is symmetric if
\begin{align}
\dual{u}\wedge\Noet_v = \dual{v}\wedge\Noet_u
\qquadtext{for all}u,v\in\Gamma TM
\label{Tensors_sym_Noet}
\end{align}
and $T^\LSym$ is symmetric if
\begin{align}
\dual{u}\wedge\tau^\LSym_v = \dual{v}\wedge\tau^\LSym_u
\qquadtext{for all}u,v\in\Gamma TM
\label{Tensors_sym_taubel}
\end{align}

To illustrate these notions, consider the premetric formulation of
electromagnetism \cite{hehl2003foundations,hehl2008maxwell} on
spacetime $\MST$ where one starts with the action
\begin{align}
\LambdaM(Z_1,\ldots,Z_\NumZ,dA)
= \tfrac12 F\wedge \calH
\label{Tensors_LambdaM_Hehl_Ob}
\end{align}
with $F=dA$ and $\calH=\CRpre(Z_1,\ldots,Z_\NumZ,F)$ linear in
$F$. The Noether current is then
\begin{align}
\Noet_v =
\tfrac12 F\wedge i_v\calH - \tfrac12 i_v F \wedge \calH
+ d i_v A \wedge\calH
\label{Tensors_EMpre_Neot}
\end{align}
This is manifestly not $U(1)$ gauge invariant. However the
Belinfante-Rosenfeld current (\ref{Phys_tau_Lsym_can})
\begin{align}
\tau^\LSym_v =
\tfrac12 F\wedge i_v\calH - \tfrac12 i_v F \wedge \calH
\label{Tensors_EMpre_Lsym}
\end{align}
is $U(1)$ gauge invariant.

When a metric is prescribed as in the spacetime model
(\ref{Tensors_F_wedge_G}) where
\begin{align}
\LambdaM(g,Z_1,\ldots,Z_\NumZ,dA)
= \tfrac12 F\wedge\star G
\label{Tensors_action_EMg}
\end{align}
with $G=\CR(g,Z_1,\ldots,Z_\NumZ,F)$ linear in $F$, the Noether
current is given by (\ref{Tensors_EMpre_Neot}) with $\calH=\star G$
and the Belinfante-Rosenfeld current is given by
(\ref{Tensors_EMpre_Lsym}) with $\calH=\star G$. In this case
the
Belinfante-Rosenfeld current
\begin{align}
\tau^\LSym_v =
\tfrac12 F\wedge i_v\star G - \tfrac12 i_v F \wedge \star G
\label{Tensors_EM_Lsym}
\end{align}
gives rise via
(\ref{Tensorsdef_T_LSym}) and (\ref{Tensors_symm}) to the
algebraically non-symmetric Minkowski stress-energy-momentum tensor
\cite{minkowski1910grundgleichungen}
density $T^\LSym$. Examples (\ref{Intro_Lag_SymMink_1}) and
(\ref{Intro_Lag_Abraham_1}) are particular cases of
(\ref{Tensors_action_EMg}).

In general, neither $T^\Noet$ nor
$T^\LSym$ possess the algebraic symmetry (\ref{Tensors_symm}) due to
the presence of background fields.  In the
particular case of the vacuum where $\calH=\star F$, then $T^\LSym$
has algebraic symmetry, while $T^\Noet$ remains non-symmetric.  In the
absence of a preferred metric one cannot even define $T^\Noet$ or
$T^\LSym$ from $\Tden^\Noet$ and $\Tden^\LSym$ respectively.

\subsection{Conservation Laws}

In the presence of Lie-symmetries of $g$ and $Z_\Zind$, both the Noether
and Belinfante-Rosenfeld stress-energy-momentum current give rise to conserved
material quantities.
\begin{theorem}
\label{thm_d_tau_Lsym}
  If $K\in\Gamma TM$ is a Killing vector field, i.e. $\Lie_K g=0$, and
  in addition the background tensor fields $\Set{\boldsymbol Z}$
  satisfy the Lie-symmetry condition $\Lie_K Z_\Zind=0$ for
  $\Zind=1,\ldots,\NumZ$, then both the Noether current and the Belinfante-Rosenfeld current are closed
\begin{align}
d \tau^\LSym_K = d \Noet_K = 0
\label{Phys_d_tauSym_K}
\end{align}
\end{theorem}

\begin{proof}
Since $\Lie_K g=0$ and $\Lie_K Z_\Zind=0$ then from (\ref{AB_AB_GatZ})
\begin{align}
d\bigg(\Amap_K\Big(\GatZ{g},g\Big)\bigg)
+
\Bmap_K\Big(\GatZ{g},g\Big)
=0
\label{Tensors_lm_AB_Lk_g}
\end{align}
and
\begin{align*}
d\bigg(\Amap_K\Big(\GatZ{Z_\Zind},Z_\Zind\Big)\bigg)
+
\Bmap_K\Big(\GatZ{Z_\Zind},Z_\Zind\Big)
=0
\end{align*}
for all $Z_\Zind$. Thus from (\ref{Phys_tau_Lsym_A}) and
(\ref{AB_Bg_D_tau_Met_Shell})
\begin{align*}
d \tau^\LSym_K
&=
d\bigg(\Amap_K\Big(\GatZ{g},g\Big)\bigg)
 + \sum_{\Zind=1}^\NumZ
d \bigg(\Amap_K\Big(\GatZ{Z_\Zind },Z_\Zind \Big)
\bigg)
\\&=
-
\Bmap_K\Big(\GatZ{g},g\Big)
-
\sum_{\Zind=1}^\NumZ\Bmap_K \Big(\GatZ{Z_\Zind },Z_\Zind \Big)
=0
\end{align*}
\end{proof}

In terms of the maps $\Tden^\Noet$ and  $\Tden^\LSym$, using
lemma \ref{lm_tau_Tden} in appendix \ref{sch_lammas},
(\ref{Phys_d_tauSym_K}) may be written in a general coordinate basis with
$\Omega=dx^1\wedge\cdots\wedge dx^n$ as
\begin{align}
\partial_a((\Tden^\LSym){}^a{}_b K^b)
=
0
\label{Phys_d_tauSym_K_coords}
\end{align}
and from (\ref{Tensors_TNv_ab})
\begin{align}
\partial_a\bigg((\Tden^\Noet)^a{}_b K^b
+
(\partial_c K^b) i_{\partial_1} \cdots i_{\partial_n}\Big(
\sum_{\zeind=1}^\Numzeta
dx^a\wedge
d x^c \wedge i_{\partial_b}  \zeta_\zeind \wedge \pfrac{\LambdaM}{(d\zeta_\zeind)}
\Big)\bigg)
=
0
\label{Phys_d_tauN_K_coords}
\end{align}
When $g$ is Lorentzian on spacetime, by Stoke's theorem both
$\Tden^\LSym$ and $\Tden^\Noet$ yield conservation laws.
An energy
conservation law follows if $K$ is timelike, a linear momentum
conservation law if there exist three independent spacelike
translational Killing vectors and an angular momentum conservation law
follows if there exists a basis of three independent spacelike Killing
vectors generating spatial rotations. Such conservation laws make no
reference to forces (stress) or torques (moments) and are valid in the
presence of smooth non-dynamical background fields. However only
$\Tden^\LSym$, being $f$-linear, deserves the appellation
stress-energy-momentum tensor.

If $\LambdaM$ does not depend on a metric $g$ the requirement that $K$
is a Killing vector may be dropped, i.e. if
$\LambdaM(Z_1,\ldots,Z_\NumZ,\zeta_1,d\zeta_1,\ldots,\zeta_\Numzeta,d\zeta_\Numzeta)$
is independent of $g$ and $\Lie_V Z_\Zind=0$ for
$\Zind=1,\ldots,\NumZ$ and $V\in\Gamma TM$, then
\begin{align}
d \tau^\LSym_V = d \Noet_V = 0
\label{Phys_d_tauSym_v}
\end{align}
However, in the absence of a metric, no physical concept of energy or momentum
exists.

If $\Lie_K g=0$ but not {\em all} background fields $\Set{\boldsymbol
  Z}$ are Lie-symmetric then
a simple generalisation of theorem \ref{thm_d_tau_Lsym} yields
\begin{equation}
\begin{aligned}
d\Noet_K
=
d\tau^\LSym_K
&=
\sum_{\Set{\Zind|\Lie_K Z_\Zind\ne 0}}
\bigg(
d\Amap_K \Big(\GatZ{Z_\Zind },Z_\Zind \Big)
+
\Bmap_K \Big(\GatZ{Z_\Zind },Z_\Zind \Big)
\bigg)
\\&=
\sum_{\Set{\Zind|\Lie_K Z_\Zind\ne 0}}
\GatZ{Z_\Zind} \TprodL \Lie_K Z_\Zind
\end{aligned}
\label{Phys_dtau_K_ne0}
\end{equation}
which in general is not equal to zero.

\subsection{Historical perspectives in Minkowski spacetime}

In the absence of background fields $\Set{\boldsymbol Z}$ in Minkowski
spacetime,
$T^{\LSym}{}_{ab}=T^{\LSym}{}_{ba}$ and from
(\ref{Phys_d_tauSym_K_coords}), $\partial_a(T^{\LSym}{}^{a}{}_{b})=0$.
It follows that the ``moment of $T^{\LSym}{}_{\nu 0}$'' is conserved.
Thus in inertial coordinates $(x^0,x^1,x^2,x^3)$ with
\begin{align*}
M^\LSym{}_{\mu\nu} = -(x_\mu T^{\LSym}{}_{\nu 0} - x_\nu T^{\LSym}{}_{\nu 0})
\end{align*}
for $\mu,\nu=1,2,3$, vanishing at spatial infinity,
one has from (\ref{Phys_d_tauSym_K_coords}) with
$K=x_\mu \partial_\nu - x_\nu \partial_\mu$
\begin{align}
\dfrac{}{x^0} \int_{\Real^3} M^\LSym{}_{\mu\nu} dx^1\wedge dx^2\wedge dx^3
=
0
\label{Phys_int_M_munu}
\end{align}
This is identified with the conservation of (orbital) angular momentum of a
field system in $\Real^3$.  This result is a direct consequence of
(\ref{Phys_d_tauSym_K_coords}) which does not require the imposition
of any algebraical symmetry.

More generally with the Killing vector fields, for $a,b=0,1,2,3$
\begin{align}
\Rot_{ab} = x_a \partial_b - x_b \partial_a
\label{Phys_def_R_ab}
\end{align}
and
\begin{align}
T^{\LSym}(R_{ab}) = \MoM^\LSym{}_{ab}{}^c \partial_c
\label{Phys_T_bel_R_ab}
\end{align}
where
\begin{align}
\MoM^\LSym{}_{ab}{}^{c}=T^{\LSym}{}^c{}_{a}x_b - T^{\LSym}{}^c{}_{b}x_a
\label{Tensors_4_MOM}
\end{align}
is the ``moment of $T^{\LSym}{}^c{}_{b}$'', one has
\begin{align}
\partial_c \MoM^\LSym{}_{ab}{}^{c}=0
\label{Tensors_4_MOM_conserved}
\end{align}
In fact (\ref{Phys_d_tauSym_K_coords}) shows that
(\ref{Phys_int_M_munu}) and (\ref{Tensors_4_MOM_conserved})
are valid in the presence of Lie-symmetric background fields
satisfying
$\Lie_{\Rot_{ab}} Z_\Zind = 0$ for all $\Zind$ and for all $\Rot_{ab}$.

In Minkowski spacetime there exists a basis of 10 Killing vector
fields generating the algebra of the Poincar\'e group. The nature of
the algebraic symmetry of $T^\LSym{}_{ab}$ depends on the Lie symmetry
of the background fields. This is a {\em consequence} of the following lemma:
\begin{lemma}
  In Minkowski spacetime, for fixed $a$ and $b$, let $\partial_a$ and
  $\partial_b$ be two commuting translational Killing vectors such
  that $\Lie_{\partial_a} Z_\Zind=0$ and $\Lie_{\partial_b} Z_\Zind=0$
  for all $\Zind$. Then $T^{\LSym}$ is algebraically partially
  symmetric if and only if $\tau_{\Rot_{ab}}$ is closed, i.e.
\begin{align}
T^{\LSym}{}_{ab}=T^{\LSym}{}_{b\,a}
\qquad\Leftrightarrow\qquad
d\tau^{\LSym}_{R_{ab}} = 0
\label{Phys_partial_sym}
\end{align}
\end{lemma}
\begin{proof}
Since $\Lie_{\partial_a} Z_\Zind=0$ and $\Lie_{\partial_b} Z_\Zind=0$
then $d\tau_{\partial_a}=0$ and $d\tau_{\partial_b}=0$. From $f$-linearity
$\tau^{\LSym}_{R_{ab}}=x_a\tau^{\LSym}_{\partial_b}-x_b\tau^{\LSym}_{\partial_a}$
hence from (\ref{Tensors_symm_alt})
\begin{align*}
d\tau^{\LSym}_{R_{ab}}
=
dx_a\wedge\tau^{\LSym}_{\partial_b}-dx_b\wedge\tau^{\LSym}_{\partial_a}
=
(T^{\LSym}_{ab} - T^{\LSym}_{ba}) dx^0\wedge dx^1 \wedge dx^2 \wedge dx^3
\end{align*}
Hence (\ref{Phys_partial_sym}).
\end{proof}
Thus total (orbital) angular momentum conservation in Minkowski
spacetime, generated by $\Set{\partial_1,\partial_2,\partial_3}$, does not
demand that $T^{\LSym}{}_{ab}$ is {\em fully} algebraically symmetric
for all $a,b$.

It is of interest to note that conservation of an $SO(3,1)$ Killing
current can be obtained directly from the Noether quantities in the
absence of background fields. It follows from (\ref{Phys_d_tauSym_K})
and (\ref{Phys_d_tauSym_K_coords}) that for all Killing vector fields
$\Rot_{ab}$ one has
\begin{align}
d \Noet_{\Rot_{ab}} = 0
\label{Phys_d_N_R}
\end{align}
and
\begin{align}
\partial_c(T^{\Noet}(R_{ab})^c)=0
\label{Phys_part_T_N_R}
\end{align}
where from (\ref{LagForm_tau_can_lin})
\begin{align*}
\Noet_{\Rot_{ab}}
&=
x_a \Noet_{\partial_b} - x_b\Noet_{\partial_b}
+
\sum_{\zeind=1}^\Numzeta
\Big(
dx_a\wedge i_{\partial_b} \zeta_\zeind \wedge
\pfrac{\LambdaM}{(d\zeta_\zeind)}
-
dx_b\wedge i_{\partial_a} \zeta_\zeind \wedge
\pfrac{\LambdaM}{(d\zeta_\zeind)}
\Big)
\end{align*}
and from  (\ref{Tensors_TNv_ab})
\begin{align*}
T^{\Noet}(R_{ab})
&=
T^{\Noet}{}^c{}_{a}x_b - T^{\Noet}{}^c{}_{b}x_a
\\&
\hspace{-2em}
+\star^{-1}\sum_{\zeind=1}^\Numzeta\Big(
dx^c \wedge dx_a \wedge i_{\partial_b} \zeta_\zeind \wedge
\pfrac{\LambdaM}{(d\zeta_\zeind)}
-
dx^c \wedge dx_b \wedge i_{\partial_a} \zeta_\zeind \wedge
\pfrac{\LambdaM}{(d\zeta_\zeind)}
\Big)
\end{align*}
Thus the Noether current $\Noet_{\Rot_{ab}}$ does not coincide with the
moment of Noether linear momentum:
$x_a \Noet_{\partial_b} - x_b\Noet_{\partial_b}$.
The additional terms result from the fact that $\Noet_v$ is not
$f$-linear in $v$.
However although both $T^{\LSym}$ and $T^{\Noet}$ give rise to conserved
quantities, the lack of $f$-linearity in $T^{\Noet}$ precludes its use
for the definition of stress over curved 2-surfaces.

Equations (\ref{Phys_d_N_R}) and (\ref{Phys_part_T_N_R}) remain valid
in the presence of Lie-symmetric background fields with
$\Lie_{\Rot_{ab}} Z_\Zind = 0$ for $\Zind=1,\ldots,\NumZ$ and fixed
$a,b$,

\subsection{The Einstein-Hilbert stress-energy-momentum tensor and its associated
currents}
\label{sch_EH}

In general relativity the variational derivative
$\displaystyle\GatZ{g}$ is used to define the algebraically symmetric
Einstein-Hilbert stress-energy-momentum tensor
$T^\Met\in\Gamma\tenbun^{[\Tform,\Tform]} M$ for the theory:
\begin{align}
T^\Met(u,v)
= 2\star^{-1} \Big(\GatZ{g}\TprodL (\dual{u}\otimes \dual{v})\Big)
\qquadtext{for}
u,v\in\Gamma T M
\label{AB_def_T_Ein_alt}
\end{align}
which is manifestly $f$-linear in $u$ and $v$.

Using an arbitrary vector field $v$ on $M$ with a metric $g$, it is
convenient to use $T^\Met$ to define the associated $(n-1)$-form
current $\tau^\Met_v\in\Gamma\Lambda^{n-1} M$ by
\begin{align}
\tau^\Met_v = \star (T^\Met({v},-))
\label{AB_def_tau_Ein}
\end{align}
which is manifestly $f$-linear in $v$. Thus
\begin{align}
T^\Met
=
(\star^{-1}\tau^\Met_{X_a})\otimes e^a
\label{AB_T_tau_Ein}
\end{align}
where $\Set{e^a}$ and $\Set{X_a}$ constitute mutually dual frames.
From (\ref{AB_def_T_Ein_alt}) it is also clear that since $g$ is a
symmetric tensor, $T^\Met$ satisfy the algebraic symmetry
\begin{align}
T^\Met(u,v)=T^\Met(v,u)
\qquadtext{for all}
u,v\in\Gamma T M
\label{Phys_T_alg_sym}
\end{align}
(c.f. (\ref{Tensors_symm})) and hence $\Set{\tau^\Met_{X_a}}$ satisfy
the {\em algebraic symmetry condition}
\begin{align}
e_a\wedge\tau^\Met_{X_b} - e_b\wedge\tau^\Met_{X_a}
=0
\label{Phys_tau_a_sym}
\end{align}
where $e_a=g(X_a,X_b) e^b$.

In lemma \ref{lm_tau_met} appendix \ref{sch_lammas} it is  shown that
\begin{align}
\tau^\Met_v = \Amap_v\Big(\GatZ{g},g\Big)
\label{AB_tau_Ein}
\qquadtext{for all}
v\in\Gamma TM
\end{align}
Thus from (\ref{AB_tau_Mat}) it follows that in terms of the map $\Amap_v$,
\begin{align}
\tau^\Met_v
=
i_v\LambdaM
-
\sum_{\Zind=1}^\NumZ
\Amap_v \Big(\GatZ{Z_\Zind },Z_\Zind \Big)
-
\sum_{\zeind=1}^\Numzeta\Big(
i_v \zeta_\zeind\wedge \pfrac{\LambdaM}{\zeta_\zeind}
+
i_v d\zeta_\zeind\wedge \pfrac{\LambdaM}{(d\zeta_\zeind)} \Big)
\label{AB_tau_Mat_res}
\end{align}

Let the  exterior covariant derivative of $\tau^\Met_v$ be
defined by
\begin{align}
(D\tau^\Met)_v = (i_v e^a) d\tau^\Met_{X_a}  - i_v d e^a \wedge \tau^\Met_{X_a}
\label{AB_def_Dtau}
\end{align}
This corresponds to the standard Levi-Civita covariant derivative in
the case of a torsion free metric-compatible connection $\nabla$ and $
(D\tau^\Met)_v =0$ implies $ \nabla\cdot T^{\Met}=0 $.

 In lemma
\ref{lm_Dtau_met} appendix \ref{sch_lammas} it is
shown that in terms of the map $\Bmap_v$,
\begin{align}
(D\tau^\Met)_v = -\Bmap_v\Big(\GatZ{g},g\Big)
\label{AB_Dtau_Ein}
\end{align}
Hence for a system on $\LambdaM$-shell one has, from
(\ref{AB_Bg_D_tau_Met_Shell}),
\begin{align}
(D\tau^\Met)_v
&=
\sum_{\Zind=1}^\NumZ\Bmap_v \Big(\GatZ{Z_\Zind },Z_\Zind \Big)
\label{AB_D_tau_Met}
\end{align}

It is worth stressing that, in general, even if $N=0$ this relation is
not a conservation law since in general $ (D\tau^\Met)_v \neq
d\tau^\Met_v $. However, if $K\in\Gamma TM$ is a Killing vector field,
$\Lie_K g=0$, then from (\ref{AB_tau_Ein}), (\ref{Tensors_lm_AB_Lk_g})
and (\ref{AB_Dtau_Ein})
\begin{align*}
d\tau^\Met_K
&=
d\bigg(\Amap_K\Big(\GatZ{g },g \Big)\bigg)
=
-
\Bmap_K\Big(\GatZ{g },g \Big)
=
(D\tau^\Met)_K
\end{align*}
i.e.
\begin{align}
(D\tau^\Met)_K = d \tau^\Met_K
\label{Dtau_K}
\end{align}

From the definitions above, the relationship between the Einstein-Hilbert,
Noether and Belinfante-Rosenfeld currents is given by
\begin{align}
\tau^\Met_v
+
\sum_{\Zind=1}^\NumZ
\Amap_v \Big(\GatZ{Z_\Zind },Z_\Zind \Big)
=
\tau^\LSym_v
&=
\Noet_v
+
d S_v
\label{AB_tauCan_tau_Com}
\end{align}

As stated in (\ref{Phys_d_tauSym_K}) $\tau^\LSym_K$ and $\Noet_K$ give
rise to conserved currents associated with each Killing vector field
$K$ in the presence of the Lie-symmetric background fields
$\Set{\boldsymbol Z}$. By contrast, from (\ref{Dtau_K}) and
(\ref{AB_Dtau_Ein}), in the presence of arbitrary background fields,
whether Lie-symmetric or not,
\begin{align}
d\tau^\Met_K
&=
\sum_{\Zind=1}^\NumZ\Bmap_K \Big(\GatZ{Z_\Zind },Z_\Zind \Big)
\label{AB_d_tau_Met}
\end{align}
which in general does {\em not} vanish. However in general, as stated above,
unlike $\tau^\Met_K$ neither $\tau^\LSym_K$ nor $\Noet_K$ possesses
the algebraic symmetry condition
(\ref{Tensors_symm}),(\ref{Phys_tau_a_sym}) in the presence of
any background fields.

If there are no background tensor fields $\Set{\boldsymbol Z}$, i.e. $\NumZ=0$
then from (\ref{AB_tauCan_tau_Com})
\begin{align}
\tau^\LSym_v=\tau^\Met_v = \Noet_v + dS_v
\qquadtext{for all}
v\in\Gamma TM
\label{Phys_tau_Lsym_can_mat}
\end{align}
and hence if $K$ is
Killing from (\ref{Phys_d_tauSym_K})
\begin{align}
d \tau^\Met_{K}
=
d \tau^\LSym_K = d \Noet_K
=
0
\label{Phys_d_tauMet_K}
\end{align}
In this case, since $\tau^\Met_v$ satisfies the algebraically
symmetry condition it follows from (\ref{Phys_tau_Lsym_can_mat}) that
$\tau^\LSym_v$ also does. However in general $\Noet_v$ does not.

In the context of model (\ref{Intro_Lag_Abraham_1}) the
Einstein-Hilbert stress-energy-momentum tensor $T^\Met$ is given by
(\ref{AB_def_T_Ein_alt}) by evaluating the derivative
$\displaystyle\GatZ{g}$  at
\begin{align}
(g,Z^{\text{de}},Z^{\text{db}},
Z^{\text{he}},Z^{\text{hb}},V,dA)
=
(g_0,Y^{\text{de}},Y^{\text{db}},
Y^{\text{he}},Y^{\text{hb}},W,dA)
\label{Phys_Y_vals}
\end{align}
where $g_0$ is an arbitrary background metric with associated Hodge
map $\star_0$ and the $4$-velocity of the
medium $W$ satisfies $g_0(W,W)=-1$. The
$Y^{\text{I}}\in\Set{Y^{\text{de}},Y^{\text{db}},Y^{\text{he}},Y^{\text{hb}}}$
are spatial with respect to $W$, i.e.
$Y^{\text{I}}\big(g_0(W,-)\big)=0$ and $i_{W} Y^{\text{I}}(\alpha)=0$
and satisfy the adjoint properties with respect to $g_0$:
\begin{align}
\begin{gathered}
\alpha\wedge\star_0 Y^{\text{de}}(\beta)=
\beta\wedge\star_0 Y^{\text{de}}(\alpha) \,,\qquad
\alpha\wedge\star_0 Y^{\text{hb}}(\beta)=
\beta\wedge\star_0 Y^{\text{hb}}(\alpha) \\
\quadand
\alpha\wedge\star_0 Y^{\text{db}}(\beta)=-
\beta\wedge\star_0 Y^{\text{he}}(\alpha)
\end{gathered}
\label{Intro_Y0_constraints}
\end{align}
for all $\alpha,\beta\in\Gamma\Lambda^1M$.  The tensor $T^\Met$ is the
Abraham stress-energy-momentum tensor for electromagnetic fields in
(moving) media \cite{dereli2007covariant,dereli2007new}.

\section{Caveats and Conclusions}

The results in this article have direct relevance to the construction
of variational formulations of field systems in spacetime where
complex microscopic interactions are represented by phenomenological
macroscopic constitutive relations between the dynamic fields $\Set{\boldsymbol
\zeta}$.  It is tacitly assumed that the variational field equations
based on the matter action $\int_M\LambdaM$ admit such $\Set{\boldsymbol
\zeta}$ as non-trivial physically acceptable solutions. This condition
may impose constraints on the constitutive modelling. For example
if
\begin{align*}
\LambdaM(g,J,A)=\tfrac12 dA\wedge\star dA + A\wedge J
\end{align*}
where $A\in\Gamma\Lambda^1 \MST$ is dynamical and $J\in\Gamma\Lambda^3 \MST$
is a background field then the $\LambdaM$ matter shell condition
$\displaystyle \frac{\delta\LambdaM}{\delta A}=0$
implies $d J=0$.
The physical
interpretation of such models for continuous media can then be
facilitated by finding {\em conserved} currents associated with continuous
symmetries.

For the class of theories described by a diffeomorphism invariant
matter action $\int_M \LambdaM$, where $\LambdaM$ in general depends
on a background gravitation field, a collection of background tensors
$\{\boldsymbol Z\}$ and matter fields $\{ \boldsymbol\zeta\}$ the
differential forms $\tau^{\Met}_v$, $\tau^\LSym$ and $\Noet_v$ are
related by (\ref{AB_tauCan_tau_Com})
when all {\em matter} fields $\Set{\boldsymbol\zeta}$ are ``on
$\LambdaM$-shell'' and from (\ref{AB_d_tau_Met}),
$d\tau^\Met_K\ne0$ in general.


However, in Einstein's theory
\begin{align}
\LambdaT
=
{\cal R}\star 1+
\LambdaM(g,Z_1,\ldots,Z_\NumZ,\zeta_1,d\zeta_1,\ldots,\zeta_\Numzeta,d\zeta_\Numzeta)
\label{Phys_Ein_Lag}
\end{align}
where $\cal R$ is the curvature scalar derived from $g$.
If the metric $g$ is dynamical as well as the matter fields
$\Set{\boldsymbol\zeta}$ one
has the additional on-shell Einstein equation
\begin{align}
{\cal G}_v = \tau^\Met_v
\label{Phys_Ein_Ein}
\end{align}
where the Einstein $(n-1)$-form ${\cal G}_v\in\Gamma\Lambda^{n-1}M$
for $v\in\Gamma TM$ is \cite{benn1987introduction}
\begin{align}
{\cal G}_v= \tfrac12 R(X_c,X_d,X_b,e^a)  e^c\wedge e^d\wedge i_v \star
(e_a\wedge e^b)
\label{Phys_Ein_def_Gv}
\end{align}
Since the Bianchi identity for ${\cal G}_v$ yields $(D{\cal G})_v=0$
for all $v\in\Gamma TM$, (\ref{Phys_Ein_Ein}) implies
\begin{align}
(D\tau^\Met)_v=0
\label{Phys_Ein_D_tau_met}
\end{align}
Then from (\ref{Dtau_K}), if $K$ is Killing one has the conservation law
\begin{align}
d\tau^\Met_K=0
\label{Phys_Ein_d_tau_met}
\end{align}
and hence from (\ref{AB_D_tau_Met}) and (\ref{Phys_Ein_D_tau_met})
\begin{align*}
\sum_{\Zind=1}^\NumZ\Bmap_v \Big(\GatZ{Z_\Zind },Z_\Zind \Big)
=0
\end{align*}
This relation may impose constraints on the
dependence of $\LambdaM$ on the non-gravitational background structure
$\Set{\boldsymbol Z}$. For example in the analysis of a gravitational
wave propagating in a material medium.


In physical applications it is often convenient to break a system into
weakly interacting subsystems in order to analyse the dynamics of
subsystems perturbatively.  However care is required in extending the
consequences of diffeomorphism invariance for $\LambdaM$ to top-forms
$\Lambda^s$ describing diffeomorphic invariant sub-systems.  For each
$\Lambda^s$ one may define an associated Noether form $\Noet^{s}_v$:
\begin{align}
\Noet^{s}_v =
i_v \Lambda^s - \sum_{\zeind=1}^\Numzeta
\Lie_v \zeta_\zeind \wedge \pfrac{\Lambda^s}{(d\zeta_\zeind)}
\end{align}
Thus
\begin{align}
\Noet_v = \sum_{s} \Noet^{s}_v
\end{align}
Similarly
\begin{align}
\tau^\Met_v = \sum_{s} \tau^{s,\Met}_v
\end{align}
where
\begin{align}
T^{s,\Met}(u,v) = 2\star^{-1}
\Big(\frac{\Delta \Lambda^s}{\Delta g}\TprodL
  (\dual{u}\otimes \dual{v})\Big) \qquadtext{for} u,v\in\Gamma
  T M
\end{align}
and
\begin{align}
\tau^{s,\Met}_v = \star (T^{s,\Met}({v},-))
\label{AB_def_tau_Ein_s}
\end{align}
When all matter fields $\Set{\boldsymbol\zeta}$ are ``on
$\LambdaM$-shell'' they will not necessarily be ``on
$\Lambda^s$-shell''and one cannot associate conserved currents with
subsystems, in general. This is simply a reflection of the interaction
between subsystems in situations where $g$ and $\Set{\boldsymbol Z}$
are background fields. However, given some decomposition
$\LambdaM=\sum_s \Lambda^s$, suppose that for some $s=s_0$ there
exists a subset of matter fields $ \{ \zeta_B \} $ such that $\displaystyle
\frac{\delta \Lambda^{s_0}} { \delta \zeta_B} =0 $ for $B=1\ldots
Q_0$.  In this situation one may regard the quantities $ \zeta_B,
\,d\zeta_B $ for $B=Q_0+1 \ldots Q$ as background fields to supplement
those in $\boldsymbol Z$. The analysis in this paper is then
applicable by disregarding all $\Lambda^s$ with $s\neq s_0$ ,
replacing $\LambdaM$ by $\Lambda^{s_0}$ and disregarding the equations
$\displaystyle \frac{\delta \Lambda^{s_0}} { \delta \zeta_B} \ne 0 $ for $B=Q_0+1
\ldots Q$.


Maps such as $T^\Noet$, $T^{\LSym}$ and $T^\Met$ have been traditionally used
to construct densities of field energy and linear momentum  in
Minkowski spacetime.  Applied to closed systems in the presence of the
Lie-symmetric background fields $\Set{g,\boldsymbol Z}$ this procedure
is strictly only meaningful for $T^{\Noet}$ and $T^{\LSym}$ since they
alone give rise to conserved currents in Minkowski spacetime.
Furthermore only $T^\LSym$ can meaningfully be used to describe stress
and angular momentum in Minkowski spacetime. In a general spacetime
with a non-dynamic metric neither $T^\LSym$, $T^\Met$ nor $T^\Noet$
give rise to conserved Killing currents unless $\Lie_Kg=0$ and $\Lie_K
Z_\Zind=0$ for all $\Zind$.

In any background metric, the model described by
(\ref{Intro_Lag_SymMink_1}) yields a particular $T^{\LSym}$ known as
the non-symmetric Minkowski stress-energy-momentum tensor and $
T^{\Met}$ as its symmetrised version.

The model (\ref{Intro_Lag_Abraham_1}) yields a symmetric $T^{\Met}$
known as the Abraham stress-energy-momentum tensor (which for a
general Lie-symmetric $\boldsymbol Z$ does not generate conserved
Killing currents). This model does however yield $T^{\LSym}$ that
coincides with the non-symmetric Minkowski stress-energy-momentum
tensor.  This sheds light on the relationship between the Abraham and
non-symmetric Minkowski stress-energy-momentum tensors.  The precise
relation between these two tensors, in this model follows from
(\ref{AB_tauCan_tau_Com}) as
\begin{equation}
\begin{aligned}
\tau^\LSym_v
&=
\tau^\Met_v
+
\Amap_v \Big(\GatZ{V},V \Big)
+
\Amap_v \Big(\GatZ{Z^{\text{de}}},Z^{\text{de}}\Big)
\\&\qquad\qquad
+\Amap_v \Big(\GatZ{Z^{\text{db}}},Z^{\text{db}}\Big)
+
\Amap_v \Big(\GatZ{Z^{\text{he}}},Z^{\text{he}}\Big)
+
\Amap_v \Big(\GatZ{Z^{\text{hb}}},Z^{\text{hb}}\Big)
\end{aligned}
\label{Conc_Ahb_Mink}
\end{equation}
where all tensors are evaluated at the point
\begin{align*}
(g,Z^{\text{de}},Z^{\text{db}},
Z^{\text{he}},Z^{\text{hb}},V,dA)
=
(g_0,Y^{\text{de}},Y^{\text{db}},
Y^{\text{he}},Y^{\text{hb}},W,dA)
\end{align*}

Conclusions
drawn from different models are directly related to the epistemology used
to describe the linear and angular momentum of light in unbounded
media described by background constitutive tensor fields in background
gravitational fields.

Thus the physical consequences of any model based on a diffeomorphism
invariant action used to describe the dynamics of matter fields in the
presence of a specified self-consistent non-dynamic background depend
not only on the action for the model but also on a choice of objects,
such as $T^{\Met}$, $T^{\LSym}$ or $T^{\Noet}$, that are adopted to
define conserved quantities, including field energy, momentum and
angular momentum. Furthermore these may only acquire physical cogency
in the presence of sufficient Lie symmetry.

\section*{Acknowledgments}
J.G. and R.W.T. are grateful to STFC and the Cockcroft Inistitute for
support (STFC ST/G008248/1).  Y.N.O. acknowledges partial support by
the German-Israeli Foundation for Scientific Research and Development
(GIF), Research Grant No.\ 1078-107.14/2009.  This work is the outcome
of discussions initiated by the author's attendence at the 475th
Wilhelm \& Else Heraeus Seminar on ``Problems and Developments of
Classical Electrodynamics'' Bad Honnef March 2011, organised by Domenico
Giulini and Volker Perlick.

\bibliographystyle{plain}
\bibliography{GTOB}


\appendix

\section{Mathematical Details of results used in the text}
\label{sch_lammas}

\begin{lemma}
\label{lm_zero}
On an $n$ dimensional manifold $M$, given the (co-vector valued) forms
$\alpha_w\in\Gamma\Lambda^{n-1}M$ and $\beta_w\in\Gamma\Lambda^{n}M$
which are `f'-linear in $w$ such that
\begin{align}
d(\alpha_w)=\beta_w
\label{AB_lm_d_f}
\end{align}
for all $w\in\Gamma TM$ with compact support then $\alpha_w=0$ and $\beta_w=0$.
\end{lemma}
\begin{proof}
Let $w=w^a X_a$, $\alpha_a=\alpha_{X_a}$ and $\beta_a=\beta_{X_a}$.
Given any subset $U\subset M$ with boundary $\partial U$
\begin{align*}
\int_{\partial U} \iota^\star (w^a \alpha_a) = \int_U w^a \beta_a
\end{align*}
where $\iota:\partial U\to M$ is the embedding.
Assume first that $w$ has support away from the boundary $\partial U$
and in an arbitrary small region then
\begin{align*}
\int_U w^a \beta_a = 0
\end{align*}
and hence $\beta_a=0$. Thus
\begin{align*}
\int_{\partial U} w^a \iota^\star(\alpha_a) = 0
\end{align*}
For all subsets $U\subset M$. Considering $w$ to have a support on a
small set about a point in the boundary implies
$\iota^\star\alpha_a=0$. Since we can choose any $\partial U$ we show
that all the components of $\alpha_a=0$ and hence $\alpha_a=0$.
\end{proof}

\begin{lemma}
\label{lm_AB}
\textup{(\ref{AB_diVLam}) } implies \textup{(\ref{AB_ivLambda_A})} and
\textup{(\ref{AB_ivLambda_B})}.
\end{lemma}
\begin{proof}
Since (\ref{AB_diVLam}) is true for all $v\in\Gamma TM$ it is true for
all $v=w\in\Gamma TM$ with compact support. Thus (\ref{AB_ivLambda_A}) and
(\ref{AB_ivLambda_B}) follow from lemma \ref{lm_zero} setting
\begin{align*}
\alpha_w
=
i_w \LambdaM - \sum_{\ZZind=0}^{\NumZ+2\Numzeta}
\Amap_w \Big(\GatZ{\Zalt_\ZZind},\Zalt_\ZZind\Big)
\qquadand
\beta_w
=
- \sum_{\ZZind=0}^{\NumZ+2\Numzeta}
\Bmap_w \Big(\GatZ{\Zalt_\ZZind},\Zalt_\ZZind\Big)
\end{align*}
\end{proof}

\begin{lemma}
\label{lm_alpha_V}
For any $\alpha\in\Gamma\Lambda^p M$ we have \textup{(\ref{Tensors_alpha_V})}
\end{lemma}
\begin{proof}
Setting
\begin{align*}
\alpha
=
\sum_{ I_1<\ldots< I_p}\alpha_{ I_1\cdots I_p}
e^{ I_1}\wedge\cdots\wedge e^{ I_p}
\end{align*}
From (\ref{Tensors_identify_form_tens}) and (\ref{Tensors_general_V})
we have
\begin{align*}
\alpha\Tprod \Vasym
&=
\Big(\frac1{p!}
\sum_{ I_1<\ldots< I_p} \sum_{\sigma\in S_p}\alpha_{ I_1\cdots I_p}
\epsilon(\sigma) e^{\sigma( I_1)}\otimes\cdots\otimes e^{\sigma( I_p)}
\Big)
\\&\qquad\qquad\Tprod
\Big(\frac1{p!}
\sum_{ J_1<\ldots< J_p} \Vasym^{ J_1\cdots J_p}    \sum_{\rho\in S_p}
\epsilon(\rho) X_{\rho( J_1)}\otimes\cdots\otimes X_{\rho( J_p)}\Big)
\\&=
\frac1{(p!)^2}
\sum_{ I_1<\ldots< I_p} \sum_{ J_1<\ldots< J_p}
\sum_{\sigma\in S_p} \sum_{\rho\in S_p}
\alpha_{ I_1\cdots I_p} \Vasym^{ J_1\cdots J_p}
\epsilon(\sigma)\epsilon(\rho)
\delta^{\sigma( I_1)}_{\rho( J_1)}\cdots \delta^{\sigma( I_p)}_{\rho( J_p)}
\\&=
\frac1{p!}
\sum_{ I_1<\ldots< I_p} \sum_{ J_1<\ldots< J_p}
\alpha_{ I_1\cdots I_p} \Vasym^{ J_1\cdots J_p}
\delta^{ I_1}_{ J_1}\cdots \delta^{ I_p}_{ J_p}
\\&=
\frac1{p!}
\sum_{ I_1<\ldots< I_p}
\alpha_{ I_1\cdots I_p} \Vasym^{ I_1\cdots I_p}
\end{align*}
while
\begin{align*}
i_\Vasym\alpha
&=
\frac1{p!}
\sum_{ J_1<\ldots< J_p} \Vasym^{ J_1\cdots J_p}
i_{X_{ J_p}}\cdots i_{X_{ J_1}}
\Big(
\sum_{ I_1<\ldots< I_p}\alpha_{ I_1\cdots I_p}
e^{ I_1}\wedge\cdots\wedge e^{ I_p}
\Big)
\\&=
\frac1{p!}
\sum_{ J_1<\ldots< J_p}\sum_{ I_1<\ldots< I_p}
\alpha_{ I_1\cdots I_p}\Vasym^{ J_1\cdots J_p}
i_{X_{ J_p}}\cdots i_{X_{ J_1}}
e^{ I_1}\wedge\cdots\wedge e^{ I_p}
\\&=
\frac1{p!}
\sum_{ J_1<\ldots< J_p}\sum_{ I_1<\ldots< I_p}
\alpha_{ I_1\cdots I_p}\Vasym^{ J_1\cdots J_p}
\delta^{ I_1}_{ J_1}\cdots \delta^{ I_p}_{ J_p}
\\&=
\frac1{p!}
\sum_{ I_1<\ldots< I_p}
\alpha_{ I_1\cdots I_p} \Vasym^{ I_1\cdots I_p}
\end{align*}

\end{proof}

\begin{lemma}
\label{lm_LagForm}
For any $\alpha\in\Gamma\Lambda^p M$, one has \textup{(\ref{AB_Omega_X_alpha})}
\end{lemma}
\begin{proof}
For $0\le q\le p$ we have since $(i_{v_1} \cdots i_{v_q} \alpha )
\wedge (i_{v_{q+1}} \cdots i_{v_p} \Omega )\in\Gamma\Lambda^n M$
\begin{align*}
\lefteqn{
(i_{v_1} \cdots i_{v_q} \alpha )
\wedge
(i_{v_{q+1}} \cdots i_{v_p} \Omega )}\qquad\qquad&
\\
&=
(-1)^{q-1} i_{v_q} (i_{v_1} \cdots i_{v_{q-1}} \alpha )
\wedge
(i_{v_{q+1}} \cdots i_{v_p} \Omega )
\\&=
(-1)^{q-1} (-1)^{(p-q)}(i_{v_1} \cdots i_{v_{q-1}} \alpha )
\wedge
i_{v_q} (i_{v_{q+1}} \cdots i_{v_p} \Omega )
\\&=
(-1)^{p-1}(i_{v_1} \cdots i_{v_{q-1}} \alpha )
\wedge
(i_{v_{q}} \cdots i_{v_p} \Omega )
\end{align*}
Hence from (\ref{Tensors_Omega_Phi_Y}) and (\ref{Tensors_alpha_V})
\begin{align*}
(\Omega\otimes \Vasym)\TprodL \alpha
&=
(\Vasym\Tprod\alpha) \Omega
=
(i_\Vasym \alpha) \Omega
=
\frac{1}{p!} \sum_{ I_1<\cdots< I_p} \Vasym^{ I_1\cdots I_p}
(i_{X_{ I_p}} \cdots i_{X_{ I_1}} \alpha )
\wedge
\Omega
\\&=
\frac{1}{p!} \sum_{ I_1<\cdots< I_p} \Vasym^{ I_1\cdots I_p}
(-1)^{p(p-1)}\alpha
\wedge
(i_{X_{ I_p}} \cdots i_{X_{ I_1}}  \Omega )
\\&=
\alpha\wedge\Big(\frac{1}{p!} \sum_{ I_1<\cdots< I_p} \Vasym^{ I_1\cdots I_p}
i_{X_{ I_p}} \cdots i_{X_{ I_1}}  \Omega \Big)
\\&=
\alpha \wedge i_\Vasym \Omega
\end{align*}
hence (\ref{AB_Omega_X_alpha}).
\end{proof}

\begin{lemma}
\label{lm_variation}
\textup{(\ref{LagForm_def_delta_frac_ea})} holds.
\end{lemma}

\begin{proof}
From (\ref{AB_def_pfrac_zeta}) and
(\ref{LagForm_def_pfrac_d_alpha_mu})
\begin{align*}
\lefteqn{
\int_M\beta\wedge\frac{\delta\LambdaM}{\delta\zeta_\zeind}}
\quad&
\\
&=
\dfrac{}{\varepsilon}\Big|_{\varepsilon=0}
\int_M
\LambdaM(g,Z_1,\ldots,Z_\NumZ,\zeta_1,d\zeta_1,\ldots,
\zeta_\zeind+\varepsilon\beta,d\zeta_\zeind+\varepsilon d\beta,
\ldots,\zeta_\Numzeta,d\zeta_\Numzeta)
\\&=
\int_M \bigg(
\beta\wedge \pfrac{\LambdaM}{\zeta_\zeind}
+
d\beta\wedge \pfrac{\LambdaM}{(d\zeta_\zeind)}\bigg)
\\&=
\int_M \bigg(
\beta\wedge \pfrac{\LambdaM}{\zeta_\zeind}
+
d\Big(\beta\wedge \pfrac{\LambdaM}{\zeta_\zeind}\Big)
- (-1)^{p_\zeind}
\beta\wedge d\Big(\pfrac{\LambdaM}{\zeta_\zeind}\Big)
\bigg)
\\&=
\int_M
\beta\wedge \bigg(\pfrac{\LambdaM}{\zeta_\zeind}
+
(-1)^{p_\zeind+1}
d\Big(\pfrac{\LambdaM}{(d\zeta_\zeind)}\Big)
\bigg)
\end{align*}
since $\beta$ has compact support. Since this is true for all $\beta$
then (\ref{LagForm_def_delta_frac_ea}) follows.
\end{proof}

\begin{lemma}
\label{lm_Lie_v_beta_alpha}
For any two forms $\alpha\in\Gamma\Lambda^{n-p} M$ and
$\beta\in\Gamma\Lambda^{p} M$ then
\begin{align}
\Lie_v \beta\wedge\alpha
=
d(i_v \beta\wedge\alpha)
+
(-1)^p i_v \beta\wedge d\alpha +i_v d\beta\wedge\alpha
\label{lms_Lie_v_beta}
\end{align}
\end{lemma}
\begin{proof}
\begin{align*}
\Lie_v \beta\wedge\alpha
&=
d i_v \beta\wedge\alpha
+
i_v d \beta\wedge\alpha
=
d(i_v \beta\wedge\alpha)
+
(-1)^p i_v \beta\wedge d\alpha +i_v d\beta\wedge\alpha
\end{align*}
\end{proof}

\begin{lemma}
\label{lm_LvZ_dA_B}
For any tensors $Z\in\Gamma\tenbun^{{\slist}}M$ and
$\Psi\in\Gamma(\Lambda^nM\otimes\tenbun^{\cnj{\slist}}M)$ then
\begin{align}
\Psi\Tprod (\Lie_v Z) = d\big(\Amap_v(\Psi,Z)\big) + \Bmap_v(\Psi,Z)
\label{AB_thm_dA_B}
\end{align}
Furthermore this decomposition is unique in the sense that if
\begin{align}
\Psi\Tprod (\Lie_v Z) = d \alpha_v + \beta_v
\label{AB_thm_dalpha_beta}
\end{align}
where $\alpha_v$ and $\beta_v$ are `f'-linear in $v$ then
\begin{align}
\alpha_v = \Amap_v(\Psi,Z)
\qquadand
\beta_v = \Bmap_v(\Psi,Z)
\label{AB_thm_alpha_beta}
\end{align}

\end{lemma}
\begin{proof}
First using (\ref{AB_def_AB_f}) we prove that (\ref{AB_thm_dA_B}) for
$f\in\Gamma\tenbun^{()}M$
\begin{align*}
\Omega\Tprod(\Lie_v f)
=
\Omega\, v(f)
=
d\big(\Amap_v(\Omega,f)\big)
+
\Bmap_v(\Omega,f)
\end{align*}
Using (\ref{AB_def_AB_alpha}) this is true for 1-forms
$\alpha\in\Gamma\tenbun^{[\Tform]}M$
\begin{align*}
(\Omega\otimes u)\Tprod(\Lie_v \alpha)
&=
\Omega\,(i_u \Lie_v\alpha)
=
(-1)^{n+1}i_u \Omega\wedge d i_v\alpha
+
(-1)^{n+1}i_u \Omega\wedge i_v d\alpha
\\&=
d (i_u \Omega\wedge i_v\alpha)
-d i_u \Omega\wedge i_v\alpha
+
(-1)^{n+1}i_u \Omega\wedge i_v d\alpha
\\&=
d\big(\Amap_v(\Omega\otimes u,\alpha)\big)
+
\Bmap_v(\Omega\otimes u,\alpha)
\end{align*}
Now using (\ref{AB_def_AB_u}) we show that  for vectors
$u\in\Gamma\tenbun^{[\Tvec]}M$
\begin{align*}
(\Omega\otimes \alpha)&\Tprod(\Lie_v u)
=
\Omega\,\big(\alpha(\Lie_v u)\big)
=
\Omega\,\big(v\big(\alpha(u)\big)\big)
-
\Omega\,(i_u \Lie_v\alpha)
\\&=
\Omega\,\big(v\big(\alpha(u)\big)\big)
-d (i_u \Omega\wedge i_v\alpha)
+d i_u \Omega\wedge i_v\alpha
-
(-1)^{n+1}i_u \Omega\wedge i_v d\alpha
\\&=
d\big(\Amap_v(\Omega\otimes\alpha,u)\big)
+
\Bmap_v(\Omega\otimes\alpha,u)
\end{align*}
Using (\ref{AB_def_AB_Ten}) it follows that (\ref{AB_thm_dA_B}) is true
across tensor products. Assuming it true for $Z_1,Z_2$ then
\begin{align*}
\lefteqn{
(\Omega\otimes \Phi_1\otimes\Phi_2)\Tprod\big(\Lie_v (Z_1\otimes Z_2)\big)
}\qquad&\\
&=
(\Omega\otimes \Phi_1\otimes\Phi_2)\Tprod(\Lie_v Z_1\otimes Z_2)+
(\Omega\otimes \Phi_1\otimes\Phi_2)\Tprod(Z_1\otimes \Lie_v Z_2)
\\&=
\Omega (\Phi_1\Tprod \Lie_v Z_1) (\Phi_2\Tprod Z_2)+
\Omega (\Phi_1\Tprod Z_1) (\Phi_2\Tprod \Lie_v Z_2)
\\&=
(\Phi_2\Tprod Z_2)(\Omega\otimes\Phi_1)\Tprod (\Lie_v Z_1) +
(\Phi_1\Tprod Z_1)(\Omega\otimes\Phi_2)\Tprod (\Lie_v Z_2)
\\&=
d\big(\Amap_v((\Phi_2\Tprod Z_2)\Omega\otimes \Phi_1,Z_1)\big) +
\Bmap_v((\Phi_2\Tprod Z_2)\Omega\otimes \Phi_1,Z_1)
\\&\quad+
d\big(\Amap_v((\Phi_1\Tprod Z_1)\Omega\otimes \Phi_2,Z_2)\big) +
\Bmap_v((\Phi_1\Tprod Z_1)\Omega\otimes \Phi_2,Z_2)
\\&=
d\big(\Amap_v(\Omega\otimes \Phi_1\otimes\Phi_2,Z_1\otimes Z_2)\big) +
\Bmap_v(\Omega\otimes \Phi_1\otimes\Phi_2,Z_1\otimes Z_2)
\end{align*}
Hence (\ref{AB_thm_dA_B}) is true for all tensors $Z$.

If (\ref{AB_thm_dalpha_beta}) is true then from (\ref{AB_thm_dA_B})
one has
\begin{align*}
d\big(\Amap_v(\Psi,Z) - \alpha_v\big)
&=
-\big(\Bmap_v(\Psi,Z) - \beta_v\big)
\end{align*}
and hence from lemma \ref{lm_zero}  (\ref{AB_thm_alpha_beta})  follows.
\end{proof}

\begin{lemma}
\label{lm_Lie_v_beta_AB}
Given $\Omega\otimes\Vasym\in\Gamma(\Lambda^n M\otimes
\tenbun^{[\Tvec,\ldots,\Tvec]}M)$ where $\Vasym$ is antisymmetric and
$[\Tvec,\ldots,\Tvec]$ has length $p$ and
$\beta\in\Gamma\Lambda^p M$ then
\begin{align}
\Amap_v(\beta,\Omega\otimes\Vasym) = i_v \beta\wedge i_\Vasym \Omega
\label{AB_forms_Av}
\end{align}
and
\begin{align}
\Bmap_v(\beta,\Omega\otimes\Vasym)
=
(-1)^p i_v \beta\wedge d i_\Vasym \Omega +
i_v d\beta\wedge i_\Vasym \Omega
\label{AB_forms_Bv}
\end{align}
\end{lemma}
\begin{proof}
From (\ref{AB_thm_dA_B}),(\ref{AB_Omega_X_alpha}) and
(\ref{lms_Lie_v_beta}) we have
\begin{align*}
d\big(\Amap_v(\beta,\Omega\otimes\Vasym)\big) +
\Bmap_v(\beta,\Omega\otimes\Vasym)
&=
(\Omega\otimes\Vasym)\TprodL(\Lie_v\beta)
=
(\Lie_v\beta)\wedge i_\Vasym\Omega
\\
&\hspace{-3em}=
d(i_v \beta\wedge i_\Vasym\Omega)
+
(-1)^p i_v \beta\wedge di_\Vasym\Omega +i_v d\beta\wedge i_\Vasym\Omega
\end{align*}
Now apply lemma \ref{lm_zero}.
\end{proof}

\begin{lemma}
\label{lm_Lie}
If the system is on $\LambdaM$-shell then \textup{(\ref{AB_tau_Mat_Shell})} and
\textup{(\ref{AB_Bg_D_tau_Met_Shell})} hold.
\end{lemma}
\begin{proof}
From (\ref{LagForm_delta_frac_ea_res})
\begin{align*}
i_v \zeta_\zeind\wedge \pfrac{\LambdaM}{\zeta_\zeind}&
+
i_v d\zeta_\zeind\wedge \pfrac{\LambdaM}{(d\zeta_\zeind)}
\\
&=
(-1)^{p_\zeind} i_v \zeta_\zeind\wedge d\Big(\pfrac{\LambdaM}{(d\zeta_\zeind)}\Big)
+
i_v d\zeta_\zeind\wedge \pfrac{\LambdaM}{(d\zeta_\zeind)}
\\&=
-
d\Big( i_v \zeta_\zeind\wedge \pfrac{\LambdaM}{(d\zeta_\zeind)}\Big)
+
d i_v \zeta_\zeind\wedge \pfrac{\LambdaM}{(d\zeta_\zeind)}
+
i_v d\zeta_\zeind\wedge \pfrac{\LambdaM}{(d\zeta_\zeind)}
\\&=
-
d\Big( i_v \zeta_\zeind\wedge \pfrac{\LambdaM}{(d\zeta_\zeind)}\Big)
+
\Lie_v \zeta_\zeind\wedge \pfrac{\LambdaM}{(d\zeta_\zeind)}
\end{align*}
Hence (\ref{AB_tau_Mat_Shell}) follows from (\ref{AB_tau_Mat}).
Likewise from (\ref{LagForm_def_delta_frac_ea}),
(\ref{LagForm_dynamical}) and $d^2=0$
\begin{align*}
(-1)^{p_\zeind} i_v \zeta_\zeind\wedge &
d\Big(\pfrac{\LambdaM}{\zeta_\zeind}\Big)
+
i_v d\zeta_\zeind\wedge\pfrac{\LambdaM}{\zeta_\zeind}
+
(-1)^{p_\zeind+1} i_v d\zeta_\zeind\wedge d\Big(\pfrac{\LambdaM}{(d\zeta_\zeind)}\Big)
\\&=
(-1)^{p_\zeind} i_v \zeta_\zeind\wedge d\Big(\frac{\delta\LambdaM}{\delta\zeta_\zeind}\Big) +
i_v d\zeta_\zeind\wedge\bigg(
\pfrac{\LambdaM}{\zeta_\zeind}
+
d\Big(\pfrac{\LambdaM}{(d\zeta_\zeind)}\Big)\bigg)
\\&=
(-1)^{p_\zeind} i_v \zeta_\zeind\wedge d\Big(\frac{\delta\LambdaM}{\delta\zeta_\zeind}\Big) +
i_v d\zeta_\zeind\wedge\frac{\delta\LambdaM}{\delta\zeta_\zeind}
\\&= 0
\end{align*}
hence (\ref{AB_Bg_D_tau_Met_Shell}) follows from (\ref{AB_Bg_D_tau_Met})
\end{proof}


\begin{lemma}
\label{lm_L_fv_alpha}
\begin{align}
\Lie_{(fv)} \alpha = f\, \Lie_v \alpha + df\wedge i_v\alpha
\quadtext{for}
f\in\Gamma\Lambda^0M\,,
v\in\Gamma TM
\textup{ and }
\alpha\in\Gamma\Lambda^pM
\label{lms_L_fv_alpha}
\end{align}
\end{lemma}
\begin{proof}
From Cartan's identity
\begin{align*}
\Lie_{(fv)} \alpha
&=
i_{(fv)} d \alpha
+
d i_{(fv)} d \alpha
=
f i_{v} d \alpha
+
d (f\, i_{v}\alpha)
\\&=
f i_{v} d \alpha
+
df\wedge i_{v}\alpha
+
f d i_{v}\alpha
=
f\, \Lie_v \alpha + df\wedge i_v\alpha
\end{align*}
\end{proof}

\begin{lemma}
\label{lm_T_tau_symmetry}
\textup{(\ref{Tensors_symm})} and \textup{(\ref{Tensors_symm_alt})}
are equivalent if $\Omega\in\Gamma\Lambda^n M$ is non-vanishing.

Furthermore if \textup{(\ref{Tensors_symm_alt})} is true for one
non-vanishing $\Omega$ it is true for all non-vanishing
$\Omega\in\Gamma\Lambda^n M$.
\end{lemma}
\begin{proof}
\begin{align*}
T(u,v)\Omega
=
g(\Tden(u),v)\Omega
=
i_{\Tden(u)} \dual{v} \, \Omega
=
\dual{v} \wedge i_{\Tden(u)} \Omega
\end{align*}
hence {(\ref{Tensors_symm})} and {(\ref{Tensors_symm_alt})}
are equivalent if $\Omega\in\Gamma\Lambda^n M$  is non-vanishing.

Substituting  $\Omega\to\hat\Omega=J\Omega$ where $J\in\Gamma\Lambda^0
M$ is non-vanishing then
\begin{align*}
\dual{v} \wedge i_{\Tden(u)} \hat\Omega-\dual{u} \wedge i_{\Tden(v)}
\hat\Omega
=
J\big(\dual{v} \wedge i_{\Tden(u)} \Omega-\dual{u} \wedge i_{\Tden(v)}
\Omega\big)
=
0
\end{align*}

\end{proof}

\begin{lemma}
\label{lm_tau_Tden}
Given a map $\tau:\Gamma TM\to \Gamma\Lambda^{(n-1)} M$,
$\tau:v\mapsto \tau_v$, let $\Tden:\Gamma TM\to\Gamma TM$ be defined
with respect to the non vanishing top form $\Omega$
via $\tau_v=i_{\Tden}\Omega$.


Using a coordinate frame $(x^1,\ldots,x^n)$ with
$\Tden(v)^a= dx^a(\Tden(v))$ and $\Omega=dx^1\wedge\cdots\wedge dx^n$ implies
\begin{align}
dx^a\wedge\tau_v=\Tden(v)^a \Omega
\label{lms_T_a_dxa}
\end{align}
\begin{align}
\Tden(v)^a = i_{\partial_n}\cdots i_{\partial_1} (dx^a\wedge\tau_v)
\label{lms_T_a_dxa_alt}
\end{align}
and
\begin{align}
d \tau_v = \partial_a(\Tden(v)^a) \Omega
\label{lms_d_tau_v}
\end{align}

Using an orthonormal coframe $\Set{e^1,\ldots,e^n}$ with
$\Tden(v)^a= e^a(\Tden(v))$ and $\Omega=\star1$ then
\begin{align}
e^a\wedge\tau_v=\Tden(v)^a \Omega
\label{lms_T_a_ea}
\end{align}
and
\begin{align}
\Tden(v)^a = \star^{-1} (dx^a\wedge\tau_v)
\label{lms_T_a_ea_alt}
\end{align}
\end{lemma}
\begin{proof}
Using the coordinate frame $\tau_v=i_{\Tden(v)}\Omega=\Tden(v)^a\,
i_{\partial_a} \Omega$
\begin{align*}
dx^a\wedge\tau_v
=
\Tden(v)^b \, dx^a\wedge  i_{\partial_b}\Omega
=
\Tden(v)^b \, \delta^a_b \Omega
=
\Tden(v)^a \Omega
\end{align*}
since clearly
\begin{align*}
dx^a\wedge  i_{\partial_b}\Omega
&=
dx^a\wedge  i_{\partial_b} dx^1\wedge\cdots\wedge dx^n
\\&=
(-1)^{b-1} dx^a\wedge
dx^1\wedge\cdots\wedge dx^{b-1}\wedge dx^{b+1}\wedge\cdots\wedge dx^n
\\&=
 dx^1\wedge\cdots\wedge dx^{b-1}\wedge dx^a\wedge
dx^{b+1}\wedge\cdots\wedge dx^n
=
\delta^a_b dx^1\wedge\cdots\wedge dx^n
\end{align*}
Hence (\ref{lms_T_a_dxa}) and since $i_{\partial_n}\cdots
i_{\partial_1} dx^1\wedge\cdots\wedge dx^n=1$
(\ref{lms_T_a_dxa_alt}).
Also
\begin{align*}
d \tau_v
&=
d(i_{\Tden(v)} \Omega)
=
d( \Tden(v)^a i_{\partial_a} \Omega)
=
d( \Tden(v)^a) \wedge i_{\partial_a} \Omega +
\Tden(v)^a \, d(i_{\partial_a} \Omega)
\\&=
\partial_b( \Tden(v)^a) \, dx^b \wedge i_{\partial_a} \Omega
=
\partial_a( \Tden(v)^a) \Omega
\end{align*}
Equation (\ref{lms_T_a_ea}) is proved similar to
(\ref{lms_T_a_dxa}) and (\ref{lms_T_a_ea_alt}) is trivial.
\end{proof}

\begin{lemma}
\label{lm_tau_met}
\textup{(\ref{AB_tau_Ein})} holds
\end{lemma}
\begin{proof}
Using a $g$-orthonormal frame $X_a$ and dual  coframe $e^a$,  set
$\eta_{ab}=g(X_a,X_b)\in\Real$ and
$T^\Met=T_{ab} e^a\otimes e^b$.  Then from (\ref{AB_def_T_Ein_alt})
\begin{align*}
T_{ab} \star 1 =
\star T^\Met(X_a,X_b)
= 2\, \GatZ{g}\TprodL (e_a\otimes e_b)
\end{align*}
hence
\begin{align*}
2\GatZ{g}
=
T^{ab} (\star1)\otimes X_a\otimes X_b
\end{align*}
From the algebraic symmetry of $T^\Met$, $T^{ab}=T^{ba}$.

From (\ref{AB_def_tau_Ein}) with $v=v^a X_a$
\begin{align*}
\tau^\Met_v=
\star (T^\Met({v},-))
=
\star  \Big((T_{ab} e^a\otimes e^b)({v},-)\Big)
=
v^a T_{ab} \star  e^b
\end{align*}
thus using (\ref{AB_def_AB_Ten}) and (\ref{AB_def_AB_alpha})
with $\Omega=\star1$ yields
\begin{align*}
2\Amap_v\Big(&\GatZ{g},g\Big)
=
\Amap_v\Big(T^{ab} (\star1)\otimes X_a\otimes X_b,
e_c\otimes e^c\Big)
\\&=
\Amap_v\Big((X_b\Tprod e^c)(\star1)\otimes (T^{ab} X_a),e_c\Big) +
\Amap_v\Big((T^{ab} X_a\Tprod e_c)(\star1)\otimes X_b,e^c\Big)
\\&=
\Amap_v\Big((\star1)\otimes (T^{ac} X_a),e_c\Big) +
\Amap_v\Big(\eta_{ac}T^{ab} (\star1)\otimes X_b,e^c\Big)
\\&=
e_c(v) i_{T^{ac} X_a} \star 1 +
e^c(v) T^a{}_c i_{X_a} \star 1
=
v_c T^{ac} \star e_a + v^c T^{a}{}_c \star e_a
\\&=
2v^c T_{ac} \star e^a
=
2\tau_v
\end{align*}
\end{proof}

\begin{lemma}
\label{lm_Dtau_met}
\textup{(\ref{AB_Dtau_Ein})} holds
\end{lemma}
\begin{proof}
  Again using an orthonormal frame and setting $T^\Met=T_{ab}
  e^a\otimes e^b$ and $v=v^a X_a$ as in lemma \ref{lm_tau_met} and
  using (\ref{AB_def_AB_Ten}), (\ref{AB_def_AB_alpha}), the algebraic
  symmetry of $T^\Met$ and (\ref{AB_def_Dtau}) yields
\begin{align*}
2\Bmap_v\Big(&\GatZ{g},g\Big)
=
\Bmap_v\Big(T^{ab} (\star1)\otimes X_a\otimes X_b,
e_c\otimes e^c\Big)
\\&=
\Bmap_v\Big((X_b\Tprod e^c)(\star1)\otimes (T^{ab} X_a),e_c\Big) +
\Bmap_v\Big((T^{ab} X_a\Tprod e_c)(\star1)\otimes X_b,e^c\Big)
\\&=
\Bmap_v\Big((\star1)\otimes (T^{ac} X_a),e_c\Big) +
\Bmap_v\Big(T_c{}^{b} (\star1)\otimes X_b,e^c\Big)
\\&=
(-1)^{n-1} T^{ac} (\star e_a) \wedge  i_v d e_c
-
d(T^{ac} \star e_a ) \wedge i_v e_c
\\&\qquad+
(-1)^{n-1} T_c{}^{b} (\star e_b)\wedge i_v d e^c
-
d(T_c{}^{b} \star e_b ) \wedge i_v e^c
\\&=
2\Big(
i_v d e^c \wedge  T_{ac} (\star e^a)
-
d(T_{ac} \star e^a ) \wedge i_v e^c\Big)
\\&=
-2\Big(
v^c d(\tau^\Met_{X_c})
-
i_v d e^c \wedge  \tau^\Met_{X_c}\Big)
=
-2(D\tau^\Met)_v
\end{align*}
\end{proof}

\def\Inc{{\textup{Inc}}}
\def\Free{{\textup{Free}}}
\def\len{{\textup{len}}}
\def\Ihat{{\hat I}}
\def\Jhat{{\hat J}}
\def\Vmul{{\cal V}}

\label{ch_Id}

The remainder of this appendix is to derive
(\ref{Tensors_T_Neot_coords}) in the text.  Since it is convenient not to use the
summation convention for multiindex summations, the following lemmas
require the introduction of some new notation.

Given positive integers $n$ and $r$, where $r\le n$ then denote
$\Inc(r,n)$ as the  set of lists of $r$ elements, which are increasing
\begin{align*}
  \Inc(r,n)=\Set{I:\Set{1,\ldots,r}\to \Set{1,\ldots,n}\,\big|\,
    I_1<I_2<\cdots<I_r}
\end{align*}
We use the notation $I\in\Inc(r,n)$ with $\len(I)=r$. For example
$\Inc(2,3)=\Set{[1,2],[1,3],[2,3]}$.
The sum over the set $\Inc(r,n)$ is then written
\begin{align*}
\sum_{I\in\Inc(r,n)}
=
\sum_{1\le I_1<I_2<\cdots<I_r\le n}
=
\sum_{I_1=1}^n \sum_{I_2=I_1+1}^n\cdots \sum_{I_r=I_{r-1}+1}^n
\end{align*}

Given positive integers $n$ and $r$, $\Free(r,n)$ is the set of lists of $r$ elements not necessarily increasing
\begin{align*}
\Free(r,n)=\Set{\Ihat:\Set{1,\ldots,r}\to \Set{1,\ldots,n}}
\end{align*}
We use the notation $\Ihat\in\Free(r,n)$. For example
\begin{align*}
\Free(2,3)=\Set{[1,1],[1,2],[1,3],[2,1],[2,2],[2,3],[3,1],[3,2],[3,3]}
\end{align*}
The sum over the set $\Free(r,n)$ is then written
\begin{align*}
\sum_{I\in\Free(r,n)}
=
\sum_{1\le I_1,\ldots,I_r\le n}
=
\sum_{I_1=1}^n \sum_{I_2=1}^n\cdots \sum_{I_r=1}^n
\end{align*}

Let $\Set{e^a,a=1,\ldots,n}$ be a coframe (not necessarily
orthonormal) and $X_a$ its dual. Let
\begin{align*}
e^I=e^{I_1}\wedge\cdots\wedge e^{I_r}
\end{align*}
and let $i_a=i_{X_a}$ and
\begin{align}
i_I=i_{I_r}\cdots i_{I_1}
\label{Id_def_iI}
\end{align}
Clearly for $\alpha_I$ antisymmetric then
\begin{align}
\sum_{I\in\Inc(r,n)} \alpha_I e^I
&=
\frac{1}{r!}\sum_{\Ihat\in\Free(r,n)} \alpha_I e^I
\label{Id_compare_sums}
\end{align}

For free list $\Ihat\in\Free(r,n)$ let
$\epsilon(\Ihat)$ is the signature of $\Ihat$ and $\Inc(\Ihat)$
is the increasing form of $\Ihat$. Here $\epsilon(\Ihat)=0$ if $\Ihat$
contains repeated indices.
Thus
\begin{align}
e^{\Ihat}=\epsilon(\Ihat) e^{\Inc(\Ihat)}
\label{Id_e_Ihat_eps}
\end{align}
For free lists $\Ihat,\Jhat\in\Free(r,n)$ let
\begin{align*}
\delta_\Jhat^\Ihat
=
\begin{cases}
\epsilon(\Ihat)\epsilon(\Jhat)
&\quadtext{if}
\Inc(\Ihat)=\Inc(\Jhat)
\\
0 &
\quadtext{Otherwise}
\end{cases}
\end{align*}
Thus for increasing lists $I,J\in\Inc(r,n)$ then
$\delta_J^I=1$ if $I=J$ and $\delta_J^I=0$ if $I\ne J$.

We use concatenation to represent the combining of lists, so that if
$\Ihat\in\Free(r,n)$ and $\Jhat\in\Free(s,n)$ then
$\Ihat\Jhat\in\Free(r+s,n)$ in the natural way, i.e.
\begin{align*}
(\Ihat\Jhat)_\mu =
\begin{cases}
\Ihat_\mu &\quadtext{if} \mu\le r
\\
\Jhat_{\mu-r} &\quadtext{if} \mu> r
\end{cases}
\end{align*}
Likewise if $a\in\Set{1,\ldots,n}$ and $I\in\Inc(r,n)$ then
$aI\in\Free(r+1,n)$.
We use the backslash to represent the removal of an element from a
list. That is for $\Jhat\in\Free(r,n)$ and $1\le s\le r$ then $J\backslash
J_s\in\Free(r-1,n)$ is given by
\begin{align*}
(\Jhat\backslash J_s)_\mu =
\begin{cases}
\Jhat_\mu &\quadtext{if} \mu< s
\\
\Jhat_{\mu+1} &\quadtext{if} \mu\ge s
\end{cases}
\end{align*}


\begin{lemma}
\begin{align}
e^a \wedge i_a \alpha = p \alpha
\qquadtext{for} \alpha=\Gamma\Lambda^p M
\label{Id_iaea}
\end{align}
\end{lemma}
\begin{proof}
Clearly true for $p=0$. By induction, assume true for
$\alpha\in\Gamma\Lambda^p M$ let $\beta\in\Gamma\Lambda^1 M$.
\begin{align*}
e^a \wedge i_a (\beta\wedge\alpha)
&=
e^a\wedge(i_a\beta\wedge\alpha)-e^a\wedge(\beta\wedge i_a\alpha)
=
\beta\wedge\alpha+ \beta\wedge e^a\wedge i_a\alpha
\\&=
(p+1)\beta\wedge\alpha
\end{align*}
\end{proof}

\begin{lemma}
\begin{align}
\sum_{I\in\Inc(r,n)} i_I(e^I\wedge\alpha)
=
\frac{1}{r!}\sum_{\Ihat\in\Free(r,n)} i_\Ihat(e^\Ihat\wedge\alpha)
=
\binom{n-p}{r}\alpha
\quadtext{for} \alpha=\Gamma\Lambda^p M
\label{Id_iIeI}
\end{align}

\end{lemma}
\begin{proof}
Clearly the first two expression are equivalent.

Consider $r=1$ then
\begin{align*}
i_a (e^a\wedge\alpha)
&=
\delta_a^a \alpha
-
e^a\wedge i_a\alpha
=
(n-p)\alpha
\end{align*}

Then the left hand side of (\ref{Id_iIeI}) becomes
\begin{align*}
\frac{1}{r!}\sum_{\Ihat\in\Free(r,n)} i_\Ihat(e^\Ihat\wedge\alpha)
&=
\frac{1}{r!}i_{\Ihat_r}\cdots i_{\Ihat_1} ( e^{\Ihat_1}\wedge\cdots\wedge
e^{\Ihat_r}\wedge\alpha)
\\&=
\frac{1}{r!}(n-(p+r-1))i_{\Ihat_r}\cdots i_{\Ihat_2} ( e^{\Ihat_2}\wedge\cdots\wedge
e^{\Ihat_r}\wedge\alpha)
\\&=
\cdots=
\frac{1}{r!}(n-(p+r-1))\cdots(n-p) \alpha
\\&=
\frac{1}{r!}(n-p-r+1)\cdots(n-p) \alpha
=
\frac{(n-p)!}{r!(n-p-r)!} \alpha
\end{align*}
\end{proof}

\begin{corrol}
If $\deg(e^I\wedge\alpha)=n$ then
\begin{align}
\sum_{I\in\Inc(r,n)} i_I (e^I\wedge\alpha)=\alpha
\label{Id_iIeI_n}
\end{align}

\end{corrol}
\begin{proof}
Since $n-p=r$.
\end{proof}

\begin{lemma}
For increasing lists $I,J\in\Inc(r,n)$ and $\Omega\in\Gamma\Lambda^n M$ then
\begin{align}
e^I\wedge i_J \Omega = \Omega \delta_J^I
\label{Id_eIiJ_Om}
\end{align}
\end{lemma}
\begin{proof}
In the proof of lemma \ref{lm_tau_Tden} we show
\begin{align*}
e^a\wedge i_b\Omega = \delta^a_b \Omega
\end{align*}
So
\begin{align*}
e^I\wedge i_J \Omega
&=
e^{I_1}\wedge\cdots\wedge e^{I_r}\wedge i_{J_r}\cdots i_{J_1} \Omega
=
e^{I_1}\wedge\cdots\wedge e^{I_{r-1}}\wedge i_{J_{r-1}}\cdots i_{J_1} \Omega
\delta^{I_r}_{J_r}
\\&=
\Omega
\delta^{I_r}_{J_r}\delta^{I_{r-1}}_{J_{r-1}}\cdots \delta^{I_{1}}_{J_{1}}
=
\Omega \delta_J^I
\end{align*}
\end{proof}

\begin{corrol}
For free lists $\Ihat,\Jhat\in\Free(r,n)$
and $\Omega\in\Gamma\Lambda^n M$ then
\begin{align}
e^\Ihat\wedge i_\Jhat \Omega = \Omega \delta_\Jhat^\Ihat
\label{Id_eIiJhat_Om}
\end{align}
\end{corrol}
\begin{proof}
\begin{align*}
e^\Ihat\wedge i_\Jhat \Omega
=
\epsilon(\Ihat)\epsilon(\Jhat)
e^{\Inc(\Ihat)}\wedge i_{\Inc(\Jhat)}
\Omega
=
\epsilon(\Ihat)\epsilon(\Jhat)
\delta^{\Inc(\Ihat)}_{\Inc(\Jhat)}
\Omega
=
\delta_\Jhat^\Ihat
\Omega
\end{align*}
\end{proof}

\begin{lemma}
Let $I\in\Inc(r,n)$ and $J\in\Inc(r+1,n)$ and
$\Omega\in\Gamma\Lambda^n M$ then
\begin{align}
e^I\wedge i_J \Omega
=
\sum_{s=1}^{r+1} (-1)^{s-1} \delta^I_{J\backslash J_s} i_{J_s} \Omega
\label{Id_eIiJ+_Om}
\end{align}
\end{lemma}
\begin{proof}
First observe that for the left hand side of (\ref{Id_eIiJ+_Om}) to be
non zero then $I\subset J$. Since $I$ and $J$ are increasing then
$J_{r+1}\ne I_s$ for $s=1,\ldots,r-1$. Thus there are two cases,
either $J_{r+1}=I_r$ or $J_{r+1}$ does not equal any $I_s$. I.e.
\begin{align}
i_{J_{r+1}} e^I = (-1)^{r-1}\delta_{J_{r+1}}^{I_{r}}e^{I\backslash I_{r}}
\label{Id_lm_inc_diff}
\end{align}
By induction on $r$. Clearly true when $I\in\Inc(0,n)$.
Assume true when $I\in\Inc(r-1,n)$ then
\begin{align*}
e^I\wedge i_J \Omega
&=
e^I\wedge i_{J_{r+1}} i_{J\backslash J_{r+1}} \Omega
\\&=
(-1)^r i_{J_{r+1}} (e^I\wedge  i_{J\backslash J_{r+1}} \Omega) +
(-1)^{r-1} i_{J_{r+1}} e^I\wedge  i_{J\backslash J_{r+1}} \Omega
\\&\hspace{18em}
\text{from (\ref{Id_eIiJ_Om}) and (\ref{Id_lm_inc_diff})}
\\&=
(-1)^r\delta^I_{J\backslash J_{r+1}} i_{J_{r+1}} \Omega +
\delta_{J_{r+1}}^{I_{r}} e^{I\backslash I_{r}}\wedge  i_{J\backslash J_{r+1}} \Omega
\\&\hspace{18em}\text{from induction hypothesis}
\\&=
(-1)^r\delta^I_{J\backslash J_{r+1}} i_{J_{r+1}} \Omega +
\delta_{J_{r+1}}^{I_{r}}
\Big(
\sum_{s=1}^{r} (-1)^{s-1}
\delta^{I\backslash I_{r}}_{J\backslash J_{r+1}\backslash J_s} i_{J_s} \Omega
\Big)
\\&=
(-1)^r\delta^I_{J\backslash J_{r+1}} i_{J_{r+1}} \Omega +
\sum_{s=1}^{r} (-1)^{s-1}
\delta^{I}_{J\backslash J_s} i_{J_s} \Omega
\\&=
\sum_{s=1}^{r+1} (-1)^{s-1}
\delta^{I}_{J\backslash J_s} i_{J_s} \Omega
\end{align*}

\end{proof}

\begin{lemma}
Let $\Vmul^J$ be a multiindex object where
$J\in\Inc(r+1,n)$ and let
$I\in\Inc(r,n)$ and $\Omega\in\Gamma\Lambda^n M$ then
\begin{align}
\sum_{J\in\Inc(r+1,n)} \Vmul^J \sum_{s=1}^{r+1} (-1)^{s-1}
\delta^{I}_{J\backslash J_s} i_{J_s} \Omega
=
\sum_{a=1}^n \Vmul^{\Inc(aI)} \epsilon(aI) i_a\Omega
\label{Id_eIiJ+_alt}
\end{align}
\end{lemma}
\begin{proof}
The only values of $J$ such that $\delta^{I}_{J\backslash J_s}=1$ are
when $J=\Inc(aI)$ for some $a$. In this case one
value of $s$ is $\delta^{I}_{J\backslash J_s}=1$ and
\begin{align*}
(-1)^{s-1} = \epsilon(aI)
\end{align*}
Hence result
\end{proof}

\begin{corrol}
Let $\Vmul^\Jhat$ be an antisymmetric multiindex object where
$\Jhat\in\Free(r+1,n)$ and let
$I\in\Inc(r,n)$ and $\Omega\in\Gamma\Lambda^n M$ then
\begin{align}
\sum_{J\in\Inc(r+1,n)} \Vmul^J \sum_{s=1}^{r+1} (-1)^{s-1}
\delta^{I}_{J\backslash J_s} i_{J_s} \Omega
=
\sum_{a=1}^n \Vmul^{aI} i_a\Omega
\label{Id_eIiJ+_alt2}
\end{align}
\end{corrol}
\begin{proof}
Follows since
\begin{align*}
\Vmul^{\Inc(aI)}\epsilon(aI)=\Vmul^{aI}
\end{align*}
\end{proof}

\begin{corrol}
Let $\Vmul^\Jhat$ and $\alpha_\Ihat$ be antisymmetric multiindex objects
where $\Jhat\in\Free(r+1,n)$ and $\Ihat\in\Free(r,n)$
and let
$\Omega\in\Gamma\Lambda^n M$ then
\begin{align}
\sum_{I\in\Inc(r,n)}
\sum_{J\in\Inc(r+1,n)}
\Vmul^J \alpha_I
e^I\wedge i_J \Omega
=
\sum_{I\in\Inc(r,n)}
\sum_{a=1}^n \alpha_I \Vmul^{aI} i_a\Omega
\label{Id_sum_eIiJ_om}
\end{align}
\end{corrol}

\begin{lemma}
\label{lm_coord_T}
\textup{(\ref{Tensors_T_Neot_coords})} holds.
\end{lemma}
\begin{proof}
From (\ref{Tensorsdef_T_Noet})
\begin{align*}
\Noet_{\partial_a}
=
i_{\Tden^\Noet(\partial_a)}\Omega
=
\Tden^\Noet{}_a{}^b i_b\Omega
\end{align*}
Let $\deg(d\zeta_\zeind)=r+1$ and
$d\zeta_\zeind=\sum_{I\in\Inc(r+1,n)} (d\zeta_\zeind)_{I} dx^I$ and
$J\in\Inc(r+1,n)$ then from (\ref{LagForm_def_pfrac_d_alpha_mu})
\begin{align*}
dx^J \wedge\pfrac{\LambdaM}{(d\zeta_B)}
&=
\dfrac{}{\varepsilon}\Big|_{\varepsilon=0}
\LambdaM(\ldots,d\zeta_B+\epsilon dx^J,\ldots)
\\&=
\dfrac{}{\varepsilon}\Big|_{\varepsilon=0}
\Sden(\ldots,(d\zeta_B)_J+\epsilon,\ldots)\Omega
=
\pfrac{\Sden}{(d\zeta_B)_J}\Omega
\end{align*}
Hence from (\ref{Id_iIeI_n})
\begin{align*}
\sum_{J\in\Inc(r+1,n)} \pfrac{\Sden}{(d\zeta_B)_J} i_J\Omega
=
\sum_{J\in\Inc(r+1,n)} i_J \Big(dx^J \wedge\pfrac{\LambdaM}{(d\zeta_B)}\Big)
=
\pfrac{\LambdaM}{(d\zeta_B)}
\end{align*}
From (\ref{LagForm_def_tau_can})
\begin{align*}
\Noet_{\partial_a}
&=
i_a \LambdaM -
\sum_{\zeind=1}^\Numzeta
\Lie_{\partial_a} \zeta_\zeind \wedge \pfrac{\LambdaM}{(d\zeta_\zeind)}
\\&=
i_a (\Sden\Omega) -
\sum_{\zeind=1}^\Numzeta \sum_{I\in\Inc(r,n)} \sum_{J\in\Inc(r+1,n)}
\Lie_{\partial_a} (\zeta_{\zeind I} dx^I) \wedge
\pfrac{\Sden}{(d\zeta_B)_J} i_J\Omega
\\&=
\Sden i_a \Omega -
\sum_{\zeind=1}^\Numzeta \sum_{I\in\Inc(r,n)} \sum_{J\in\Inc(r+1,n)}
\partial_a (\zeta_{\zeind I})
\pfrac{\Sden}{(d\zeta_B)_J}   dx^I \wedge i_J\Omega
\\&=
\delta^b_a\Sden\, i_b \Omega -
\sum_{\zeind=1}^\Numzeta \sum_{I\in\Inc(r,n)}
\partial_a (\zeta_{\zeind I})
\pfrac{\Sden}{(d\zeta_B)_{bI}}   i_b\Omega
&\hspace{-2em}\text{from (\ref{Id_sum_eIiJ_om})}
\end{align*}
Hence
\begin{align*}
\bigg(\Tden^\Noet{}_a{}^b -
\delta^b_a\Sden +
\sum_{\zeind=1}^\Numzeta \sum_{I\in\Inc(r,n)}
\partial_a (\zeta_{\zeind I})
\pfrac{\Sden}{(d\zeta_B)_{bI}}
\bigg)
i_b\Omega
=0
\end{align*}
Thus
\begin{align*}
0 &=
\bigg(\Tden^\Noet{}_a{}^b -
\delta^b_a\Sden +
\sum_{\zeind=1}^\Numzeta \sum_{I\in\Inc(r,n)}
\partial_a (\zeta_{\zeind I})
\pfrac{\Sden}{(d\zeta_B)_{bI}}
\bigg)
e^c\wedge i_b\Omega
\\&=
\bigg(\Tden^\Noet{}_a{}^b -
\delta^b_a\Sden +
\sum_{\zeind=1}^\Numzeta \sum_{I\in\Inc(r,n)}
\partial_a (\zeta_{\zeind I})
\pfrac{\Sden}{(d\zeta_B)_{bI}}
\bigg)\delta^c_b
\Omega
\end{align*}
giving
\begin{align*}
\Tden^\Noet{}_a{}^b
=
\delta^b_a\Sden -
\sum_{\zeind=1}^\Numzeta \sum_{I\in\Inc(r,n)}
\partial_a (\zeta_{\zeind I})
\pfrac{\Sden}{(d\zeta_B)_{bI}}
\end{align*}
Now since
\begin{align*}
d\zeta_B
=
\sum_{I\in\Inc(r,n)}d(\zeta_{BI} dx^I)
=
\sum_{I\in\Inc(r,n)}\partial_b \zeta_{BI} dx^{bI}
=
\sum_{I\in\Inc(r,n)}\epsilon(bI) \partial_b \zeta_{BI} dx^{\Inc(bI)}
\end{align*}
then
\begin{align*}
\pfrac{\Sden}{(d\zeta_B)_{bI}}
=
\epsilon(bI)\pfrac{\Sden}{(d\zeta_B)_{\Inc(bI)}}
=
\epsilon(bI)^2 \pfrac{\Sden}{(\partial_b\zeta_{BI})}
=
\pfrac{\Sden}{(\partial_b \zeta_{BI})}
\end{align*}
and
\begin{align*}
\Tden^\Noet{}_a{}^b
=
\delta^b_a\Sden -
\sum_{\zeind=1}^\Numzeta \sum_{I\in\Inc(r,n)}
\partial_a (\zeta_{\zeind I})
\pfrac{\Sden}{(\partial_b \zeta_{BI})}
\end{align*}
I.e. (\ref{Tensors_T_Neot_coords}).

\end{proof}

\section{Computational method for calculating   Gateaux  derivatives}
\label{sch_Examples}

In any local frame $\{ X_a\}$ with dual co-frame $  \{e^b\}  $   let
$\LambdaM(Z , \ldots)\in\Gamma\Lambda^n M$ depend on the tensor $Z=Z^{a}{}_{b}
X_a\otimes e^b$. Relative to any chosen top-form $\Omega$  write $\LambdaM(Z, \ldots)=\Sden(Z^a_b, \ldots)\Omega$ where
$\Sden\in\Gamma\Lambda^0 M$ and $\Omega$ is
independent of $Z$. Thus $\Sden(Z^a{}_b, \ldots)$ depends on the $n^2$ variables $Z^{a}{}_b  $ with $n^2$
derivatives $\pfrac{\Sden}{Z^a{}_b}$.  In this case
\begin{align*}
\GatZ{Z}\TprodL Y
&=
\dfrac{}{\varepsilon}\Big|_{\varepsilon=0} \LambdaM(Z+\varepsilon Y, \ldots)
=
\dfrac{}{\varepsilon}\Big|_{\varepsilon=0}
\Sden(Z^a{}_b + \varepsilon Y^{a}{}_b , \ldots) \Omega
=
\pfrac{\Sden}{Z^a{}_b}  Y^{a}{}_b \Omega
\\&=
\pfrac{\Sden}{Z^a{}_b} \Omega\otimes e^a\otimes X_b \TprodL Y
\end{align*}
hence
\begin{align*}
\GatZ{Z} = \pfrac{\Sden}{Z^a{}_b} \Omega\otimes e^a\otimes X_b
\end{align*}

For example if $\Sden$ is such that $\LambdaM(Z,  \alpha,\beta,u,v)=Z(\alpha,v)
Z(\beta,u)\Omega$ for any 1-form fields  $\alpha,\beta\in\Gamma\Lambda^1 M$
and vector fields $u,v\in\Gamma T M$
then
\begin{align*}
\Sden(Z^a{}_b\alpha,\beta,u,v)=Z^a{}_b Z^c{}_d \alpha_a v^b \beta_c u^d
\end{align*}
and it follows from the definition that
\begin{align*}
\GatZ{Z}
&=
\alpha_a v^d \beta_c u^b Z^c{}_d
\Omega\otimes e^a\otimes X_b
+
\alpha_c v^d \beta_a u^b Z^c{}_d
\Omega\otimes e^a\otimes X_b
\\&=
Z(\beta,u)
\Omega\otimes \alpha\otimes v
+
Z(\alpha,v)
\Omega\otimes \beta\otimes u
\end{align*}

\end{document}